%% file: main.tex
\documentclass[11pt]{article}
\usepackage[utf8]{inputenc}
\usepackage{amsmath}
\usepackage{amssymb}
\usepackage{graphicx}
\usepackage{amsthm}
\newtheorem{theorem}{Theorem}

\newtheorem{definition}{Definition}

\newtheorem{lemma}{Lemma}
\newtheorem{example}{Example}
\newtheorem{conjecture}{Conjecture}
\newtheorem{remark}{Remark}
\newtheorem{problem}{Problem}
\newtheorem{assumption}{Assumption}
\usepackage{algpseudocode}
\usepackage{algorithmicx}
\usepackage{algorithm}

\newcommand{\beq}{\begin{equation}}
\newcommand{\eeq}{\end{equation}}
\newcommand{\barr}{\left[\begin{array}}
\newcommand{\earr}{\end{array}\right]}

\newcommand{\rank}{\mbox{rank}\,}

\newcommand{\bpf}{\begin{proof}}
\newcommand{\epf}{\end{proof}}

\newcommand{\ftwo}{\ensuremath{\mathbb{F}_{2}}}

\newcommand{\ff}{\ensuremath{\mathbb{F}}}

\newcommand{\PCId}{\mbox{PCI}(d)}

\begin{document}
\title
{Polynomial Complexity of Inversion of sequences and Local Inversion of Maps} 
\author{Virendra Sule\\Professor (Retired)\\Department of Electrical Engineering\\
Indian Institute of Technology Bombay\\Mumbai 400076, India\\vrs@ee.iitb.ac.in}
\maketitle

\input{Abstract}

\noindent
\emph{Subject Classification}: cs.CR, cs.DM, cs.CC
\emph{Keywords}: Sequences over finite fields, Inversion of a sequence, Linear Complexity, Maximal Order Complexity, Local Inversion of maps, Polynomial Recurrence Relations.
\input{S1Introduction}
\input{S2EquationsforRR}

\input{S3ComplexityofInversion}
\input{S4Condnsforspecialpoly}
\input{S5ExpectedComplexity}

\input{Conclusions}

\input{References}
\end{document}

%% file: Abstract.tex
\begin{abstract}
    This Paper defines and explores solution to the problem of \emph{Inversion of a finite Sequence} over the binary field, that of finding a prefix element of the sequence which confirms with a \emph{Recurrence Relation} (RR) rule defined by a polynomial and satisfied by the sequence. The minimum number of variables (order) in a polynomial of a fixed degree defining RRs is termed as the \emph{Polynomial Complexity} of the sequence at that degree, while the minimum number of variables of such polynomials at a fixed degree which also result in a unique prefix to the sequence and maximum rank of the matrix of evaluation of its monomials, is called \emph{Polynomial Complexity of Inversion} at the chosen degree. Solutions of this problems discovers solutions to the problem of \emph{Local Inversion} of a map $F:\ftwo^n\rightarrow\ftwo^n$ at a point $y$ in $\ftwo^n$, that of solving for $x$ in $\ftwo^n$ from the equation $y=F(x)$. Local inversion of maps has important applications which provide value to this theory. In previous work it was shown that minimal order \emph{Linear Recurrence Relations} (LRR) satisfied by the sequence known as the \emph{Linear Complexity} (LC) of the sequence, gives a unique solution to the inversion when the sequence is a part of a periodic sequence. This paper explores extension of this theory for solving the inversion problem by considering \emph{Non-linear Recurrence Relations} defined by a polynomials of a fixed degree $>1$ and satisfied by the sequence. The minimal order of polynomials satisfied by a sequence is well known as non-linear complexity (defining a Feedback Shift Register of smallest order which determines the sequences by RRs) and called as \emph{Maximal Order Complexity} (MOC) of the sequence. However unlike the LC there is no unique polynomial recurrence relation at any degree. This problem of inversion using polynomial recurrence relations is the main proposal of this paper (and its predecessors using LC) and has missed attention of researchers in the past. This paper proposes the conjecture that the complexity of inversion using polynomials of degree higher than $1$ in average cases of sequences is of polynomial order in the length of the sequence. 
\end{abstract}

%% file: S1Introduction.tex
\section{Introduction}
This papers defines and addresses the problem of inversion of a sequence, that of finding a prefix to a sequence which conforms with the rule satisfied by the sequence. Consider a sequence $S(M)=\{s_0,s_1,\ldots,s_{(M-1)}\}$ of length $M$ where elements $s_i$ belong to the Cartesian space $\ftwo^n$. Then \emph{inversion} of $S(M)$ is the problem of finding a prefix $s_{(-1)}$ (also called as \emph{inverse}) such that the extended sequence $\{s_{(-1)},s_0,s_1,\ldots,s_{(M-1)}\}$ satisfies the same rule as that satisfied by $S(M)$.
Hence to define the problem of inversion it is necessary to specify the class of rules to be considered for satisfaction by the sequence and which provide unique solution to the inverse for each of these rules. If $S(M)$ is considered a periodic sequence of period $M$, rule of \emph{Recurrence Relations} (RRs) defined by linear recurrence relation $s_{(M+k)}=s_k$ can be defined for $S(M)$ for which the unique inverse $s_{-1}=s_{(M-2)}$ exists. Given a scalar sequence $S(M)$, a polynomial $f(X_0,X_1,\ldots,X_{(m-1)})$ (in $m$-variables and of a fixed degree $d\leq m$) with co-efficients in $\ftwo$ defines the RRs
\beq\label{RR}
s_{(m+j)}=f(s_j,s_{(j+1)},\ldots,s_{(j+m-1)})
\mbox{ for } j=0,1,2,\ldots,M-m-1
\eeq
If $S(M)$ satisfies these relations then $f$ is said to be an \emph{associated} polynomial of $S(M)$. For $n>1$, a polynomial $f$ is said to be an associated polynomial for the vector sequence $S(M)=S(i)(M)$ if $f$ is associated polynomial for each of the co-ordinate sequences $S(i)(M)$. 

Let the set of all polynomials over $\ftwo$ in $m$-variables and degree $d$ with zero constant term be denoted by $P(m,d)$. All polynomials in $P(m,d)$ can thus be represented as a linear combination of monomials
\[
f=\sum_{1\leq i_1\leq i_2\leq\ldots i_r\leq m,1\leq r\leq d}a(i_1,i_2,\ldots,i_r)\mu(i_1,i_2,\ldots,i_r)
\]
with co-efficients $a(i_1,i_2,\ldots,i_r)$ in $\ff$ and monomials $\mu(i_1,i_2,\ldots,i_r)=X_{i_1}X_{i_2}\ldots X_{i_r}$ where product is over all distinct choices of $r$-variables for $r=1$ to $r=d$. If $f$ is such a polynomial expression in $m$ variables of degree $d$, the RRs (\ref{RR}) for $S$ define a linear system of equations with co-efficients of $f$ as unknowns. Let these equations be denoted as
\beq\label{RRincoeff}
s_{m+j}=f(s_j,s_{(j+1)},\ldots,s_{(j+m-1)},\bar{a}(i_1,i_2,\ldots,i_r))
\eeq
where $\bar{a}(.)$ denote the tuple of unknown co-efficients of a possible associated polynomial $f$.
Hence solutions of these co-efficients from equations (\ref{RRincoeff}) determine all associated polynomials $f$ of $S$ in $P(m,d)$. Given a vector sequence $S(M)=\{S(i)(M)\}$, the set of all polynomials $f$ which are simultaneous solutions to RRs (\ref{RRincoeff}) for each of the co-ordinate sequences, are the rules under which we shall seek to define the inverse of $S(M)$. We call the number of variables $m$ in an associated polynomial the \emph{order of the recurrence relations} and $d$ the degree. 

Hence for a fixed degree $d$, all orders $m$ of associated polynomials of $S$ can be computed from the solutions of (\ref{RRincoeff}). For $d=1$ the smallest order of a RR is known as the \emph{Linear Complexity} (LC) of the sequence. It can be shown using the well known theory of minimal polynomials of sequences that associated linear polynomial at degree $d=1$ of order equal to LC exist uniquely for periodic sequences. Hence for a fixed $d>1$, the smallest order $m$ at which there exists an associated polynomial of $S$ is a non-linear generalisation of LC. Without fixing the degree of the polynomial defining the RRs, the smallest order of RRs at which an associated polynomial exists, has been known as \emph{Maximal Order Complexity} (MOC). Since we have restricted this paper to the binary field $\ftwo$ we shall only consider the set $P(m,d)$ for $1\leq d\leq m$ and all monomials of polynomial expression in $P(m,d)$ shall be product terms of variables since the monomials will be evaluated only over the binary field.

We can now define the notion of inverse for a vector sequence. First we consider a vector prefix $S_{(-1)}$ to be the inverse for a vector sequence $S(M)$, if each co-ordinate $S_{(-1)}(i)$ is an inverse of the corresponding co-ordinate sequence $S(i)$ with a common associated polynomial $f$ across all $S(i)$.

\begin{definition}[Inverse of a sequence] Given a scalar sequence $S(M)$ for $n=1$ over $\ftwo^n$, if there exist a unique prefix element $s_{(-1)}$ in $\ftwo$ and an associated polynomial $f$ in $P(m,d)$ (with co-efficient vector $\bar{a}$ satisfying (\ref{RRincoeff}) and has a unique solution to $s_{(-1)}$ from the relation 
\beq\label{InverseRelationwithcoeff}
s_{(m-1)}=f(s_{(-1)},s_0,s_1,\ldots,s_{(m-2)},\bar{a})  
\eeq
then $s_{(-1)}$ is called the \emph{inverse} of $S(M)$ with $f$ as the associated polynomial. The inverse of a vector sequence $S(M)$ over $\ff^n$ for $n>1$, is the unique vector $S_{(-1)}$ such that each co-ordinate $S_{(-1)}(i)$ is the unique inverse for each co-ordinate sequence $S(i)$ with the condition that there exists an associated polynomial $f$ in $P(m,d)$ which satisfies the RRs (\ref{RRincoeff}) and (\ref{InverseRelationwithcoeff}) for each co-ordinate sequence $S(i)$ simultaneously. 
\end{definition}

\begin{remark}\label{Condonfforinverse}
    Uniqueness of the inverse of $S$ for an associated polynomial $f$ is determined by the relation (\ref{InverseRelationwithcoeff}) in which the co-effcients $\bar{a}$ of the associated polynomial $f$ appear. Hence this condition of unique solvability of the inverse in (\ref{InverseRelationwithcoeff}) is an additional constraint on the parameters $\bar{a}$ of the associated polynomial to have a unique inverse apart from the constraints of RRs.
\end{remark}

Formally we can now define the problem of inversion of a vector sequence $S(M)$.

\begin{problem}[Sequence Inversion Problem]
    Given a vector sequence $S(M)$ of length $M$ determine whether or not a unique inverse $s_{(-1)}$ exists for some associated polynomial $f$ in $P(m,d)$, develop computational framework to determine existence and computation of the inverse, the smallest number of variables $m$ and degree $d$ of associated polynomials which define the RRs and satisfy the condition of unique solution $s_{(-1)}$ to (\ref{InverseRelationwithcoeff}). Define the notion of complexity of inversion of a sequence and explore methods of computation. 
\end{problem}

\subsection{Complexity of Inversion}
The set of all associated polynomials $f$ of a sequence $S$ of chosen order $m$ and degree $d$ can be computed from the RRs (\ref{RRincoeff}). Further the subset of such $f$ for which a unique inverse exists can be computed using the additional equation (\ref{InverseRelationwithcoeff}). Hence these two systems of equations jointly are necessary and sufficient conditions on parameters of $f$ (with chosen $m$ and $d$) under which a unique inverse exists for each of the associated $f$. Hence an appropriate notion for the \emph{Polynomial Complexity of Inversion} can be defined in terms of these systems of equations. 

The well known notion of LC of a sequence $S$ is the smallest number $m$ of variables in a linear associated polynomial for linear RRs (\ref{RR}). It is also known to be the degree of the minimal polynomial of $S$. The co-efficient matrix of the RRs (\ref{RRincoeff}) is known as the Hankel matrix defined by the sequence. When the unique inverse of the sequence is also determined by the linear recurrence, the co-effcient $a_0$ of the first variable $X_0$ of the linear associated polynomial satisfies $a_0\neq 0$. Further, the minimality of $m$ (which is equivalent to the uniqueness of the minimal polynomial or the co-efficient vector of the linear associated polynomial) implies that the rank of the Hankel matrix of RRs is maximal and is equal to $m$ itself. Hence the LC is also an appropriate notion of complexity of inversion. LC is $m$ equal to both, the maximal rank of the Hankel matrix as well as the minimal order of linear RRs which captures all the linearly independent RRs. This notion of complexity of inversion by LC, is the motivation for defining a generalisation termed as polynomial complexity of inversion at degree $d$. However when the degree $d$ of associated polynomials is $>1$ the uniqueness of the associated polynomials is not guaranteed since the rank of the Hankel matrix of RRs $r(m)>m$. Hence the following definition of complexity appropriately captures the higher degree situation. 

\begin{definition}[Polynomial Complexity of Inversion]\label{DefPCofI}
The \emph{Polynomial Complexity of Inversion at degree $d$} of a given vector sequence $S(M)$ of $n$-tuples over $\ftwo$, is the smallest number $m$ of variables $1\leq d\leq m$ such that 
\begin{enumerate}
    \item There exists a solution to a polynomial $f$ in $P(m,d)$ satisfying the equations of RRs (\ref{RRincoeff}) for each of the co-ordinate sequence $S(i)$.
    \item \emph{Maximal Rank condition}. The rank of the matrix of RRs (\ref{RRincoeff}) as linear equations in co-effcients $\bar{a}$ of $f$, of all co-ordinate sequences taken together is maximal with respect to $m$ at the sequence length $M$. 
    \item There is a solution $f$ of RRs of all co-ordinate sequences $S(i)$ which has unique solution $s_{(-1)}(i)$ to the the inversion relation (\ref{InverseRelationwithcoeff}) for each of the co-ordinate sequences $S(i)$.
\end{enumerate} 
We denote this complexity as \PCId.
\end{definition}

\begin{remark}\emph{
    The notion of $\PCId$ is a non-linear generalisation of the LC which is equivalent to complexity of inversion at $d=1$. The non-linear complexity $\PCId$ is thus bounded above by LC}.
\end{remark}

\begin{remark}[Maximal rank condition]. The maximal rank condition in the definition is clarified by following notation. Let the system of linear equations of RRs (\ref{RRincoeff}) in the co-efficient vector $\bar{a}$ of $f$ for a scalar co-ordinate sequence $S(i)$ of a vector sequence $S(M)$ be denoted as 
\beq\label{LineqofRRs}
\bar{s}_i=H_i(m,M)\bar{a}
\eeq
where $\bar{s}_i$ is the vector of sequence elements 
\[
\bar{s}_i=[S(i)^T_m,S(i)^T_{(m+1)},\ldots,S(i)^T_{(M-1)}]^T
\]
for the $i$-th co-ordinate sequence $S(i)$ while $H_i(m,M)$ is the matrix of co-efficients of the linear system (\ref{RRincoeff}) for $S(i)$. Let $\hat{H}(m,M)$ denote the matrix 
\[
\hat{H}(m,M)=[H_1(m,M)^T,\ldots,H_n(m,M)^T]^T
\]
Then the second condition of maximal rank is that
\[
\rank \hat{H}(m,M)=\mbox{\rm max}_r \hat{H}(r,M)
\].
\end{remark}

\begin{remark}\emph{
If the two equations (\ref{RR}) and (\ref{InverseRelationwithcoeff}) are taken together for a scalar sequence $S$, there are two unknowns, the inverse $s_{(-1)}$ and the vector of co-efficients $\bar{a}(.)$ of the polynomial $f$. Hence PCI at degree $d$ is the smallest number $m$ such that the matrix of the system of equations (\ref{RRincoeff}) for co-effcients of $f$ has maximal rank, their solutions exist and such that there exists a unique solution to (\ref{InverseRelationwithcoeff}) for $s_{(-1)}$ for each of the associated polynomials. However due to the field being binary, $s_{(-1)}$ appears affine in $f$. The co-efficients of $f$ are determined by RR constraints (\ref{RRincoeff}), while the condition of unique solvability of $s_{(-1)}$ from (\ref{InverseRelationwithcoeff}) for $f$ can be guaranteed by constraining the multiplier of $s_{(-1)}$ in equations (\ref{InverseRelationwithcoeff}) to $1$ at the evaluation of $f$ at $S$. For inversion of a vector sequence the matrix of the combined system of equations (\ref{RRincoeff}) for all co-ordinate sequences required to have maximal rank and unique solvability of (\ref{InverseRelationwithcoeff}) for all co-ordinates $S_{(-1)}(i)$ of the inverse vector is required to hold for each of the co-ordinate sequences $S(i)$. }
\end{remark}

The complexity of inversion $\PCId$ is qualified at a fixed degree $d$ of an associated polynomial. While in comparison to MOC, where there is no restriction on the degree of the polynomial defining the RRs, it follows that 
\[
\mbox{MOC}\leq\PCId
\]
at any degree $d$. For $d=1$ one can specifically consider the LC and the value of the smallest order of a linear RR which determines a unique inverse. We can legitimately call this as LC of Inversion (LCI). Since LC equals the degree of the minimal polynomial of the sequence, it follows from uniqueness of the minimal polynomial that for $d=1$ the associated polynomial $f$ is unique and corresponds to a unique inverse when the sequence is periodic. It can be observed considering periodicity of the sequence that the corresponding LC of inversion (LCI) is equal to LC for periodic sequences and we have the bound
\[
\PCId\leq\mbox{LC}=\mbox{LCI}
\]

\subsection{Relation with local inversion of maps}
Problem of inversion of a sequence encompasses another important problem, that of local inversion of a map in finite fields. Let $F:\ff^n\rightarrow\ff^n$ be a map (mapping each point $x$ in $\ff^n$ to a unique point $y=F(x)$) in the image of $F$. Given $y$, a \emph{local inverse} of $F$ at given $y$ is a point $x$ in $\ff^n$ such that $y=F(x)$. In practical situations where such maps arise, a polynomial or logical representation of $F$ is not given or available or practically convenient to work with, although it always exists. Hence the action of the map $F$ is called as a Black Box operation since an algorithm or a computer program is available to compute $y$ given $x$. Solving for $x$ may be possible in many ways, for instance by a brute force search of $x$ in $\ff^n$ such that $F(x)$ equals a given $y$. Such a search is of exponential order in $n$. This can be carried out systematically by the Time Memory Tradeoff (TMTO) in square root number of operations of the brute force search. Another way to solve for a local inverse is to formulate algebraic equations for the equality $y=F(x)$ from the description of the internal logical components of $F$ and solving for the unknown $x$ by algebraic elimination of variables. This is also a known hard problem of computation. Formulation of such algebraic equations involves a large number of latent variables hence solving such equations has never been feasible. Using the notion of RRs is another attractive way to solve the local inversion problem as a problem of inversion of sequences. Consider the iterative sequence
\beq\label{IterativeSeq}
V(F,y,M)=\{y,F(y),F^{(2)}(y),\ldots,F^{(M-1)}(y)\}
\eeq
If $V_{(-1)}$ is an inverse of this sequence then $V_{(-1)}$ is a local inverse iff $F(V_{(-1)})=y$. Hence when $(F,y)$ generate a periodic sequence of iterations it also satisfies RRs of order $N$ equal to the period and has a unique polynomial $f$ defining a RR of order $m$ (of degree $1$) equal to the degree of the minimal polynomial and gives a unique inverse of the map at $y$ which is also the sequence inverse. Hence this is precisely the inverse computed by sequence inversion. The complexity $\PCId$ for the degree $d=1$, is thus the Linear Complexity (LC) of inversion of the sequence and the map. This observation has an important implication to Cryptanalysis problems since all such problems are local inversion problems of maps. The map $F$ associated with a cryptographic algorithm can be modeled using the polynomial $f$ which defines the RRs for the sequence $V(F,y,M)$. It follows that the complexity of local inversion of $F$ at $y$ can be considered as PCI of the sequence $V(F,y,M)$ at a prescribed degree and searching over increasing degree. The discovery of RRs of a sequence obtained through the Black Box operations of $F$ also provides interpolating conditions to model the internal mechanism of the map $F$. Hence such a model plays an important role of simplifying the algebraic model of the algorithm in Cryptanalysis problems. Formulation of the local inversion problem of a map $F$ as a sequence inversion problem generalises the notion of LC of inversion of the map to non-linear complexity of inversion. Such a step has the advantage that for $d>1$, $\PCId$ is expected to be much smaller than LC and close enough to MOC which is known to be of $O(\log M)$ order on the average. Hence when the average case complexity $\PCId$ is small enough, the key recovery in cryptanalysis is expected to be feasible on the average. 

\subsection{Sequences modeled by non-singular Feedback Shift Registers}
A feedback shift register (FSR) defined by a polynomial feedback function $f(X_0,X_1,\ldots,X_{(n-1)})$ defines a map $\Phi_f:\ff^n\rightarrow\ff^n$
\[
\Phi_f(x_0,x_1,\ldots,x_{(n-1)})\rightarrow (x_1,x_2,\ldots,f(x_0,x_1,\ldots,x_{(n-1)}))
\]
of register length $n$. Thus every polynomial in the RRs (\ref{RR}) of a sequence $S(M)$ of order $m$ represents an FSR of length $m$. An FSR is said to be \emph{non-singular} if the map $\Phi_f$ is invertible. A non-singular FSR representation of $S(M)$ if it exists, thus solves the inversion problem for any initial sequence $(x_0,x_1,\ldots,x_{(m-1)})$. However this condition is too strong for solving the sequence inversion problem because all that is needed is to find an inverse of a given sequence, not inverses of all sequences of same register length $m$. Let the smallest order of a non-singular FSR represenation at degree $d$ of the polynomial of a sequence be called \emph{non-singular FSR complexity} (NFSRC) at degree $d$. Let this complexity be denoted as $\mbox{NFSRC}(d)$. Then it follows that at degree $d$, 
\[
\PCId\leq\mbox{NFSRC}(d)
\]
To determine the non-singular FSR of degree $d$, the polynomial $f$ must satisfy the conditions (\ref{RR}) and the condition for non-singularity of the map $\Phi_f$. Such a condition, known as Golomb's condition, is well known \cite{Golomb} and is practically convenient, but works unfortunately only for the binary field. However Golomb's condition for non-singularity of the FSR does not involve the degree of the polynomial $f$. Hence we may consider the smallest order of the non-singular FSR irrespective of the degree as the \emph{Non-Singular FSR Complexity} (NFSRC). Solving for a non-singular FSR representation of a sequence automatically solves the inversion problem. Hence there is no need to consider complexity of inversion separately. Then we have
\[
\PCId\leq\mbox{NFSR}\leq\mbox{NFSRC}(d)
\]

\subsection{Previous work}
Sequences satisfying linear recurrence relations have been studied since a long time \cite{LempelZiv, Rueppel}. Often a nonlinear recurrence relation is available for a process with an initial condition and it is required to solve for the sequence, a problem called solution problem of recurrence equations. On the other hand discovering recurrence relations from a given initial sequence is an inverse problem. Past work exists over a long time on both of these problems \cite{Massey, LiuML, LeylaWinterhof, Nied,Nied21}. Nonlinear recurrence relations were defined for characterising the complexity of sequences by \cite{Jansen, Jansen91, JansenBoe}. Surprisingly the problem of inversion of sequences appears to be invincible in previous literature on sequences. The well known Golomb's condition for non-singularity of an FSR \cite{Golomb} is one of the rare instance of reference to inversion of an FSR map. However this condition is much stronger than required to find the inverse of an initial condition sequence of an FSR which generates the given sequence as an output sequence. This is because, an FSR feedback polynomial satisfying this condition gives the inverse of the whole initial state of the register for arbitrary initial state. For a given sequence the set of polynomials defining RRs which satisfy Golomb's condition will hence be much smaller than the set of all polynomials which define RRs and have a unique $s_{(-1)}$ term.      

\subsubsection{Complexity of sequences and cryptanalysis}
Cryptanalysis of pseudo-random (PR) generators has been a main motivation for studying RRs of their output sequences. Linear recurrence relations and LC were suggested in \cite{Rueppel} for prediction of the output sequences further than given initial part. On the other hand study of complexity of finite sequences has been of independent interest since long \cite{LempelZiv, Liuetal}. MOC and entropy of sequences were suggested in \cite{Jansen, Jansen91, JansenBoe, Peng} for studying complexity of FSR representation of sequences. Connections between MOC, LC and correlation measure of sequences and its connections with Walsh-Fourier transform has also been well studied \cite{Nied, LeylaWinterhof, Blahut, Golomb2}. It is inherently implied by the high complexity of the sequence that it's long term prediction is computationally infeasible. Hence complexity of the sequence is considered as security assurance of the PR generator for practical use. However assurance of security is based on another fundamental problem of cryptanalysis, that of key recovery from the information on the public channel. For this reason security assurance based on complexity of output sequences is not strong enough. 

Pseudo-random (PR) sequences or output streams generated by stream ciphers with non-linear operations and FSRs have been investigated for their LC profiles using the Berlekamp Massey attack \cite{Massey}. High LCs of specially generated sequences have been proved to show that such an attack is not feasible to predict the entire PR sequence from a short initial segment. However it is important to understand that prediction of further terms of a PR sequence is just one of the issues in cryptanalysis not the central one. The important issue in cryptanalysis is to determine the seed, initial vectors, symmetric and private keys of a cryptographic algorithm. Such a problem known as the key recovery problem is addressed only by the problem of inversion of the iterative sequence generated by the map or local inversion of maps. Output sequences of PR generators are not same as iterative sequences of maps at a given value in the image. Hence none of the previous studies on complexity of output sequences of PR generators address the key recovery problem. For instance even if a RRs for an output PR sequence of an encryption algorithm are known it does not allow computation of the symmetric keys or seeds without solving the inversion problem. The problem of inversion of sequences on the other hand, is of independent interest for Computational Sciences apart from addressing the key recovery problem. The local inversion problem for a map over a finite domain $F:X\rightarrow X$ is studied as the search problem by the TMTO attack algorithm \cite{Hellman,HongSarkar}. TMTO however does not use any structure and functionality of finite fields and polynomials. Hence approach to local inversion as a sequence inversion is a mathematically much sophisticated approach. 

\subsubsection{Local inversion and its application to Cryptanalysis}
Importance of the local inversion problem to solving the key recovery in cryptanalysis was pointed out in \cite{Sule1,Sule2,Sule3,Sule4}. It is shown that local inversion in principle also solves the problem of reversing RSA encryption and recovering the private key both without factoring the modulus. It is also shown how local inversion can in principle solve the discrete log problem over finite fields (DLP) as well as over elliptic curves (ECDLP). Hence local inversion provides an alternative methodology for cryptanalysis of both symmetric as well as public key cryptography. This methodology using black box computation of maps involved can be used for estimating the computational efforts and complexity of solving these known hard problem. A detailed study of density of sequences with low values of PCI for small degrees such as $2,3$ has not been studied previously and is of importance in Cryptanalysis of both symmetric ciphers and public key algorithms such as RSA and ECDLP. In previous papers referred above, local inversion was shown to be possible by using only the linear RRs and LC. In this paper we extend this method to non-linear RRs defined by polynomials and PCI.

\subsubsection{Cryptanalysis using partially available iterative sequence}\label{PartialSequence}
Cryptanalysis using inversion of a sequence also poses another fundamental issue. In a practical key recovery problem there is given a map $F$ and an output $y$ of $F$. The inversion of $F$ at $y$ is obtained from inversion of the iterative sequence (\ref{IterativeSeq}) which is given partially only upto length $M$ which is of polynomial size in $n$ (the length of unknown input). On the other hand the iterative sequence has periodic extension of length (period) exponential in $n$. This periodic sequence is the complete sequence generated by $F$ and $y$. Hence it is a crucial question how much probable it is that the inverse computed from a small partial sequence is the correct inversion of the complete sequence. In this paper we briefly address this practical issue which allows relating the local inversion of the complete sequence to that of the local inversion of the partial sequence of length $M$ and present conjectures for want of any theoretical justification. These conjectures need to be verified by carrying out computations on case studies of sequences arising in realistic problems. However this task is too massive to be within the scope of this paper hence shall be carried out in independent investigations.

Local inversion of a map is also the central objective of the TMTO attack on maps. While TMTO attack is one of the well studied algorithms \cite{Hellman, HongSarkar, Hays, BiryukovShamir} the algorithm is based on very basic operations and steps to minimize memory and time steps required in the brute force search and results in a square root order complexity relative to the domain size. The sequence inversion and its use for local inversion is based on algebraic properties of finite fields and polynomials. Hence the theory of local inverse presented in this paper is a fresh new approach using polynomial complexity and takes previous results ahead.  

%% file: S2EquationsforRR.tex
\section{Conditions for associated polynomials and existence of unique inverse}
In this section we shall formulate the conditions required to be satisfied by all associated polynomials of order $m$ and degree $d$ for existence of a unique inverse. To proceed further we make one simplifying assumption or a restriction on the class of polynomials to be considered for defining the RRs as follows.

\begin{assumption}[Weak homogeneity]
A polynomial $f(X_0,X_1,\ldots,X_{(m-1)})$ considered for defining a recurrence relation for a sequence is assumed to be \emph{weakly homogeneous} i.e. satisfies 
\[
f(0,0,\ldots,0)=0
\]
\end{assumption}
The assumption restricts $f$ to those whose constant terms are zero. The assumption is mainly required for simplification of expressions and reduces complexity numbers by one. 

Another assumption introduced earlier which we emphasise is that, the finite field is restricted to be the binary field $\ftwo$ only. For more general fields results derived in this paper are yet to be developed. Moreover, in most of the practical applications of local inversion of maps even those defined by modular operations or over general finite fields as well as symmetric ciphers the sequences can be represented over the Cartesian space $\ftwo^n$. Hence the theory of sequence inversion and local inversion of maps over $\ftwo$ has ample applications. Hence we formally state

\begin{assumption}
    Problems of Inversion of Sequences and Local Inversion of Maps are studied over the binary field $\ff=\ftwo$ only. 
\end{assumption}

\subsection{Analysis of associated polynomials}
In order to understand the equations governing the associated polynomials we need to analysis the structure of representation of these polynomials in terms of their parameters. We shall introduce a more refined representation of polynomials in $P(m,d)$ to suit the equations (\ref{RRincoeff} and the condition for unique solution of (\ref{InverseRelationwithcoeff}) than considered in the Introduction section and considering the speciality of the field to be $\ftwo$.

\subsubsection{Number of co-efficients in polynomials defining recurrence}
A general weakly homogenous polynomial in $m$ variables of degree $d$ is thus of the form
\[
\begin{array}{lcl}
f(X_0,X_1,\ldots,X_m) & = & L_1(X_0,\ldots,X_{(m-1)})+L_2(X_0,\ldots,X_{(m-1)})+\ldots\\
 & & +L_d(X_0,\ldots,X_{(m-1)})
 \end{array}
\]
where $L_k$ are homogeneous polynomials or forms of degree $k$. For the binary field $\ftwo$ the polynomials $f$ are simplified to have smallest number of terms in each of the homogeneous forms $L_k$. Number of terms in each $L_k$ over $\ftwo$ of order $m$ is
\[
n(k)=\binom{m}{k}
\]
Hence the total number of terms and so also the number of co-efficients in $f$ over $\ftwo$ of degree $d$ is
\[
n_c(f,d,m)=\sum_{j=1}^{d}n(j)=\sum_{j=1}^{d}\binom{m}{j}
\]

\subsubsection{Number of equations satisfied by associated polynomials}
The number of equations which constrain the associated polynomials $f$ are dependent on the length $M$ of the sequence, the order $m$ of recurrence but not on the degree $d$. First consider a scalar sequence $S(M)$ with $n=1$ over $\ftwo$. The number of RR constraints given by equations (\ref{RRincoeff}) on $f$ are 
\[
n_r(M,m)=(M-m)
\]
to these equations we need to add the constraint arising from requirement of unique solution of $s_{(-1)}$ to (\ref{InverseRelationwithcoeff}) 
\[
f(s_{(-1)},s_0,s_1,\ldots,s_{(m-2)},\bar{a}(.))=s_{(m-1)}
\]
This imposes further restrictions on the parameters of $f$ for existence of an inverse. These are considered next.

\subsubsection{Presentation of polynomials}
To understand the constraints satisfied by associated polynomials along with constraints to have solution for a unique inverse it is useful to consider special representation of polynomials in $P(m,d)$. 
A polynomial $f$ can always be presented as sum of monomial terms denoted as follows. For numbers $i_k$ in $[0,(n-1)]$ and $n$ variables $X_0,X_1,\ldots,X_{(n-1)}$, a $k$ degree monomial $m(A)$ for indices $i_1<i_2<\ldots<i_k$, $0\leq i_k\leq (n-1)$ is defined by
\[
m(i_1,\ldots,i_k)=X_{i_1}X_{i_2}\ldots X_{i_k}
\]
for $i_1<i_2<\ldots<i_k$. The indexing of variables defines a lexicographical order $X_i<X_j$ for $i<j$. An order on monomials $m(A)$ is defined as
\[
m(i_1,i_2,\ldots,i_r)\leq m(j_1,j_2,\ldots,j_l)\mbox{ iff }r\leq l\mbox{ and }i_k\leq j_k
\]
If $r>l$ or $i_k>j_k$ for some $k$ then $M(i_1,\ldots,i_r)>m(j_1,\ldots,j_l)$.

A polynomial $f$ in $P(m,d)$ is then presented as a sum of monomials $m(A)$ of increasing order.
\[
\begin{array}{lcl}
f(X_0,\ldots,X_{(n-1)}) & = & \sum_{A}a(A)m(A)\\
& = & c+\sum_{i}a(i)X_i+\sum_{i<j}a(i,j)X_{i}X_{j}+\\
 & & \mbox{higher order terms}
\end{array}
\]
with $c$ as the constant term as the coefficient of the empty (or zero degree) monomial $1$.

When a polynomial expression $f$ with number of variables $m$ and degree $1\leq d\leq m$ is considered for describing RRs of a sequence, the polynomial is denoted as $f(X_0,\ldots,X_{(m-a)},\bar{a})$ where $\bar{a}$ denotes the vector of unknown co-efficients of $f$ of monomial terms of increasing order. After evaluation of variables of $f$ at the sequence terms $(s_k,s_{k+1},\ldots,s_{(k+m-1)}$ we get the affine form
\[
f(S,\bar{a})=c+\sum_{A} a(A)\prod s_k\ldots s_{(k+m-1)}
\]
with the constant the term $c$ as the co-efficient of the zero degree monomial $1$. The co-efficient vector of $f$ presented as sum of monomials of increasing order is denoted as $\bar{a}$.

\subsubsection{Nature of equations satisfied by associated polynomials}
We first consider the case of scalar sequence $S(M)$. The equations of RRs (\ref{RRincoeff}) are linear in the co-efficients of $f$ while the equation (\ref{InverseRelationwithcoeff}) has bilinear terms in co-efficients of $f$ and the unknown inverse $s_{(-1)}$ along with linear terms in co-efficients of $f$. This can be observed as follows. Hence we consider a more refined representation of the polynomial expression $f$ in two terms than just as sum of monomials
\beq\label{PolynforRR}
\begin{array}{lcl}
f(X_0,X_1,\ldots,X_{(m-1)},\bar{a},\bar{b}) & = & X_0h(X_1,\ldots,X_{(m-1)},\bar{a})\\ & & +g(X_1,\dots,X_{(m-1)},\bar{b})
\end{array}
\eeq
where $\bar{a}$ and $\bar{b}$ denote the tuples of co-efficients of polynomials $h$ and $g$ respectively when $h$ and $g$ are represented as sums of monomials. Together they constitute the co-efficients of $f$.

The weak homogeneity assumption on $f$ implies that $g(0,0,\ldots,0,\bar{b})=0$ for all $\bar{b}$. Then the evaluation of $f$ at $\{s_{(-1)},S\}$ as in the (\ref{InverseRelationwithcoeff}) is
\[
f(s_{(-1)},s_0,\dots,s_{(m-2)},\bar{a},\bar{b})=s_{(-1)}h(s_0,\ldots,s_{(m-2)},\bar{a})+g(s_0,\ldots,s_{(m-2)},\bar{b})
\]
Thus components of $\bar{a}$ appear affine linearly in the function $h$ and that of $\bar{b}$ appear linearly in $g$. Hence the first term $s_{(-1)}h$ is bilinear in $s_{(-1)}$ and $\bar{a}$ while the second term $g$ is linear in $\bar{b}$ when the terms $\{s_0,s_1,\ldots,s_{(m-2)}\}$ are substituted in $h$ and $g$. Hence $h(s_0,s_1,\ldots,s_{(m-2)})$ and $g(s_0,s_1,\ldots,s_{(m-2)})$ can be expressed as scalar products,
\beq\label{vhvg}
\begin{array}{lcl}
h(s_0,s_1,\ldots,s_{(m-2)},\bar{a}) & = & <v_h(S),\bar{a}>\\
g(s_0,s_1,\ldots,s_{(m-2)},\bar{b}) & = & <v_g(S),\bar{b}>
\end{array}
\eeq
where $v_h(S)$ and $v_g(S)$ depend on the initial sequence. The equations satisfied by $f$ for RRs (\ref{RRincoeff}) are denoted as
\beq\label{eqnRR}
H(m,d)
\barr{c}\bar{a}\\\bar{b}\earr
=H_1(S)\bar{a}+H_2(S)\bar{b}=
\bar{s}
\eeq
where $\bar{s}=[s_m,s_{(m+1)},\ldots,s_{(M-1)}]^T$, $\bar{a}$, $\bar{b}$ are co-efficients of polynomials $h$ and $g$ respectively.

\begin{definition}[Hankel matrix $H(m,d)(S)$]\label{HankelMatrix}
    For a polynomial expression $f$ with parameters $\bar{a}$, $\bar{b}$ in the form (\ref{PolynforRR}), we call the matrix 
    \[
    H(m,d)(S)=[H_1(S),H_2(S)]
    \]
    of co-efficients of the linear equations (\ref{eqnRR}) the \emph{Hankel matrix} of $f$, defined by $m,d$ and monomials of $f$ evaluated at the scalar sequence $S(M)$. 
\end{definition}

\begin{remark}\label{Constancyofeval}
    It is important to observe that once $m$ and $d$ are chosen, a general polynomial expression $f$ represented in (\ref{PolynforRR}) has a unique set of possible monomial terms giving the unique Hankel matrix $H(m,d)(S)$ and vectors $v_h(S)$ and $v_g(S)$ evaluated at the given $S$. In other words the Hankel matrix determines $f$ from RRs and only depends on $m,d$ and $S$. Further for a fixed $m$, $H(m,d_2)(S)$ contains all the columns of $H(m,d_1)(S)$ for $d_2\geq d_1$.
\end{remark}

\subsection{Fundamental lemma}
Using the expression (\ref{PolynforRR}) for $f$ in terms of two functions $h$ and $g$ with their co-efficient parameters $\bar{a}$ and $\bar{b}$ respectively, necessary and sufficient conditions for existence of associated polynomials and solution of an inverse for associated polynomials are described by the following lemma.

\begin{lemma}[Fundamental Lemma]\label{FundLemma} Let $S(M)$ be a scalar sequence, $f$ be a polynomial expression in $P(m,d)$ with co-efficients $\bar{a}$, $\bar{b}$ respectively of functions $h$ and $g$ in (\ref{PolynforRR}). Then $f$ is an associated polynomial with a unique inverse $s_{(-1)}$ iff 
\begin{enumerate}
    \item Solutions $\bar{a}$, $\bar{b}$ of equations of RRs (\ref{eqnRR}) 
    \[
    H(m,d)(S)\barr{c}\bar{a}\\\bar{b}\earr=H_1(S)\bar{a}+H_2(S)\bar{b}=\bar{s}
    \]
    and the equation
    \beq\label{condonh}
    <v_h(S),\bar{a}>=1
    \eeq
    exist.
    \item The unique inverse determined by $f$ (defined by any solution $\bar{a}$, $\bar{b}$ is 
    \beq\label{Inverse}
    s_{(-1)}=s_{(m-1)}+<v_g(S),\bar{b}>
    \eeq
    \end{enumerate}
    The set of all associated polynomials in $P(m,d)$ and the inverses to $S(M)$ defined by them are determined this way. 
\end{lemma}

\begin{proof}
By definition an associated polynomial $f$ satisfies the recurrence relations RRs (\ref{RR}). For a polynomial with representation (\ref{PolynforRR}) in terms of parameters $\bar{a}$, $\bar{b}$ these equations are precisely (\ref{eqnRR}). Further, a unique inverse is determined by an associated $f$ iff $h(s_0,\ldots,s_{(m-2)})=1$ which is the equation (\ref{condonh}). The equation (\ref{InverseRelationwithcoeff}) in terms of an inverse $s_{(-1)}$ with the constraint on $h$ gives the inverse as in (\ref{Inverse}). This is explained below by reproducing the equation (\ref{InverseRelationwithcoeff}) for convenience,  
\[
    \begin{array}{lcl}
    s_{(m-1)} & = & f(s_{(-1)},s_0,s_1,\ldots,s_{(m-2)},\bar{a},\bar{b})\\
     & = & 
     s_{(-1)}h(s_0,s_1,\ldots,s_{(m-2)},\bar{a})+g(s_0,s_1,\ldots,s_{(m-2)},\bar{b})
     \end{array}
\]
    Hence this condition for unique solvability of $s_{(-1)}$ implies that
    \[
    h(s_0,s_1,\ldots,s_{(m-2)},\bar{a})=1\Leftrightarrow <v_h(S),\bar{a}>=1
    \]
    The formula for inverse then follows from (\ref{InverseRelationwithcoeff}) by evaluation of $f$ on $S$ and using the representation of $f$ in (\ref{PolynforRR}).

    Conversely if (\ref{eqnRR}) has a solution $\bar{a}$, $\bar{b}$ such that $<v_h,\bar{a}>=1$, then the polynomial function defined by 
    \[
    f(X_0,X_1,\ldots,X_{(m-1)})=X_0h(X_1,\ldots,X_{(m-1)},\bar{a})+g(X_0,\ldots,X_{(m-1)},\bar{b})
    \]
    satisfies the RRs (\ref{eqnRR}) and $h$ satisfies $h(s_0,\ldots,s_{(m-2)},\bar{a})=1$. Hence there exists a unique solution $s_{(-1)}$ of the (\ref{InverseRelationwithcoeff}) for $\bar{b}$ defining $g$ and is given by (\ref{Inverse}). Hence all associated polynomials and inverses to $S(M)$ defined by them are determined by equations (\ref{eqnRR}) and (\ref{condonh}).
\end{proof}

\subsubsection{Projection representation of $\bar{b}$}
We next consider the condition of the fundamental lemma for $f$ to be an associated polynomial with an inverse and show a decomposition of the equations defining the parameters $\bar{a}$, $\bar{b}$. 

Let $f$ in $P(m,d)$ be an associated polynomial of $S$ and determines a unique inverse. Then $f$ with parameters $\bar{a}$ and $\bar{b}$ of the polynomials $h$, $g$ in (\ref{PolynforRR}) satisfies equations (\ref{eqnRR}) and (\ref{condonh}). Hence there exists an elementary matrix $E$ over $\ftwo$ such that 
\beq\label{Projection}
E
\barr{ll}
H_1(S) & H_2(S)\\v_h & 0
\earr
\barr{c}\bar{a}\\\bar{b}\earr
=
\barr{ll}
H_{11}(S) & H_{12}(S)\\0 & H_{22}(S)
\earr
\barr{c}\bar{a}\\\bar{b}\earr
\eeq
such that $H_{11}(S)$ has full row rank. Let
\[
E\barr{c}\bar{s}\\1\earr=
\barr{c}\bar{s}_1\\\bar{s}_2\earr
\]
the condition of the fundamental lemma for $f$ to be an associated polynomial with a unique inverse corresponds to existence of a solution $\bar{a}$, $\bar{b}$ to the system
\beq\label{Constraintonh}
\barr{ll}
H_1(S) & H_2(S)\\v_h & 0
\earr
\barr{c}\bar{a}\\\bar{b}\earr
=
\barr{c}\bar{s}\\1\earr
\eeq

\begin{lemma}[Projection representation of $\bar{b}$]\label{LemProjRep}
Projection of the solution of parameters $\bar{a}$, $\bar{b}$ of an associated polynomial $f$ in $P(m,d)$ on the parameters $\bar{b}$ is given by the solution $\bar{b}$ of the linear system 
\beq\label{Projonb}
H_{22}(S)\bar{b}=\bar{s}_2
\eeq
\end{lemma}

\begin{proof}
    The elementary matrix $E$ in (\ref{Projection}) is the matrix corresponding to row elementary transformations which will transform the matrix
    \[
    \barr{ll}
    H_1(S) & H_2(S)\\v_h & 0
    \earr
    \]
    into row echelon form. Since there exists a solution to the parameters $\bar{a}$, $\bar{b}$, the number of LI rows in $H_{11}(S)$ is at most equal to the number of columns or the size of $\bar{a}$ and the equation (\ref{Projonb}) also has a solution $\bar{b}$. Then since $H_{11}(S)$ has full rank there exists a solution $\bar{a}$ to the equation
    \beq\label{Solnofparbara}
    H_{11}(S)\bar{a}+H_{12}(S)\bar{b}=\bar{s}_1
    \eeq
    for any $\bar{b}$. Since $E$ is an elementary matrix, the set of all solutions of the equation (\ref{Constraintonh}) is the same as the set of solutions of 
    \beq\label{ReducedeqnbyE}
    \barr{ll}
    H_{11}(S) & H_{12}(S)\\0 & H_{22}(S)
    \earr
    \barr{c}\bar{a}\\\bar{b}\earr
    =
    \barr{c}\bar{s}_1\\\bar{s}_2\earr
    \eeq
    This proves that the solution of (\ref{Projonb}) is the projection of the solution $\bar{a}$, $\bar{b}$ satisfying (\ref{Constraintonh}).
\end{proof}

\subsection{Conditions for associated polynomials to determine unique inverses}
Fundamental lemma gives necessary and sufficient conditions for a polynomial $f$ to be an associated polynomial of $S$ and determine a unique inverse. For a fixed degree $d$ and order $m$ there may exist several associated polynomials in $P(m,d)$ each defining an inverse. All such polynomials are determined from the representation (\ref{PolynforRR}) by solving for the indeterminate parameters $\bar{a}$, $\bar{b}$ from equations of RRs (\ref{RRincoeff}) and the inversion condition (\ref{condonh}). We now develop further conditions for existence of this class of all associated polynomials in $P(m,d)$ to determine unique inverse.

Consider again a given finite scalar sequence $S(M)=\{s_0,s_1,\ldots,s_{(M-1)}\}$ of length $M$ where the individual elements $s_i$ belong to $\ftwo$. Let a polynomial $f$ in $m$ variables is to be found such that it defines a RR for $S(M)$ i.e.\ $f$ satisfies (\ref{RRincoeff}) and there exists a solution to $s_{(-1)}$ of the equation (\ref{InverseRelationwithcoeff}). ($f$ will be always assumed to satisfy the weak homogeneity assumption). However inverses $s_{(-1)}$ determined by different $f$ can be different.  

\begin{theorem}\label{ThInverse}
    There exist associated polynomials $f$ in $P(m,d)$ for the sequence $S(M)$ defined by parameters $\bar{a}$, $\bar{b}$ in representation (\ref{PolynforRR}) and determine a unique inverse $s_{-1}$ iff the matrix
    $H_{22}(S)$ in the projection representation (\ref{Projonb}) satisfies
    \beq\label{Projrankcond}
    \rank H_{22}(S)=\rank [H_{22}(S),\bar{s}_2]
    \eeq
    If $\bar{b}$ is any solution of (\ref{Projonb}) then the unique inverse for $f$ defined by $\bar{b}$ is given by
    \beq\label{TheInverse}
    s_{(-1)}=s_{(m-1)}+<v_g(S),\bar{b}>
    \eeq
    for all solutions $\bar{a}$ of (\ref{Solnofparbara}).
\end{theorem}

\begin{proof}
    First condition is equivalent to the solvability of $\bar{a}$, $\bar{b}$ in equations (\ref{ReducedeqnbyE}) which is the same as the first two conditions of the fundamental lemma, lemma 1, or the solvability of equations (\ref{Constraintonh}). Then from lemma \ref{LemProjRep} the solutions of $\bar{b}$ are the solutions of the system of equations (\ref{Projonb}). For every solution $\bar{b}$ of (\ref{Projonb}) all associated polynomials $f$ are defined by $\bar{b}$ together with solutions of parameters $\bar{a}$ of the equation (\ref{Solnofparbara}). The formula for inverse is same as established in the fundamental lemma and depends on the value $<v_g(S),\bar{b}>$. Here $v_g(S)$ gets defined by the terms of the polynomial $g$ evaluated at the initial part of the sequence $S$. This value is precisely $<v_g(S),\bar{b}>$. This proves the first part of the statement of the theorem.
 \end{proof}

\subsubsection{Equivalence class of polynomials for common inverse}
We call associated polynomials $f_1$, $f_2$ in $P(m,d)$ for which inverses of $S(M)$ exist, to be \emph{equivalent with respect to the inversion}, if the inverses determined by the equation (\ref{InverseRelationwithcoeff}) for each of $f_1$, $f_2$ are the same for $S(M)$. This is an equivalence relation on the set of associated polynomials having inverse and since the field is binary, partitions this set in two subsets since $s_{(-1)}$ and its complement are the only two possibilities for an inverse. For a vector sequence each co-ordinate sequence imposes conditions for defining the equivalence class of $f$. Hence it is of interest to determine conditions under which there exists a unique common inverse for all those polynomials in $P(m,d)$ which are associated polynomials of the vector sequence $S(M)$.

The set of all associated polynomials $f$ in $P(m,d)$, for which unique inverse exist are defined by solutions of parameters $\bar{a}$, $\bar{b}$. However as noted in Remark \ref{Constancyofeval}), the Hankel matrix $H(f,S)$ and evaluations of $h$, $g$ the vectors $v_h(S)$, $v_g(S)$ are fixed by $m$ and $d$ for the given sequence $S$. Hence all associated $f$ in Theorem \ref{TheInverse}  have a common inverse iff for each of the solutions $\bar{b}$ of (\ref{Projonb}) the function $g$ evaluated at $S$ is constant (either exclusively $0$ or $1$). This condition is captured in the next theorem. 

\begin{theorem}\label{ThCommonInverse}
All associated polynomials $f$ defined by solutions of co-effcients $\bar{b}$ of (\ref{Projonb}) and corresponding solutions $\bar{a}$ of (\ref{Solnofparbara}), have the same common inverse iff the matrix $H_{22}(S)$ satisfies the following condition with respect to the vector $v_g(S)$ 
    \beq\label{Constancycond}
    \rank H_{22}(S)=\rank
    \barr{c}H_{22}(S)\\v_g(S)\earr
    \eeq
\end{theorem}

\begin{proof}
    The formula for inverse for an associated polynomial $f$ defined by parameters $\bar{a}$ and $\bar{b}$ is as given in (\ref{TheInverse}). Hence the inverse $s_{(-1)}$ is only dependent on parameters $\bar{b}$ of $f$ in terms of the value of the function $g(S)$ which equals $<v_g(S),\bar{b}>$. Hence the inverse is common across all associated polynomials iff 
    \[
    <v_g(S),\bar{b}>=\mbox{ constant}
    \]
    Since all solutions of $\bar{b}$ are that of the equation (\ref{Projonb}), these form the set
    \[
    \{\bar{b}=b_0+\ker H_{22}(S)\}
    \]
    where $b_0$ is one solution of (\ref{Projonb}). Hence $<v_g(S),\bar{b}>$ is constant equal to $<v_g(S),b_0>$ iff $<v_g(S),x>=0$ for all $x$ in $\ker H_{22}(S)$. This implies that the vector $v_g(S)$ is in the span of rows of the matrix $H_{22}(S)$ which is expressed as condition (\ref{Constancycond}).
\end{proof}

This section was devoted to developing necessary and sufficient conditions for polynomials in $P(m,d)$ to be associated polynomials of a scalar sequence $S(M)$ and to determine a unique inverse. Similarly these conditions are extended in Theorem (\ref{ThCommonInverse}) such that the set of all associated polynomials of $S(M)$ in $P(m,d)$ determine a common inverse. In the next section we consider the problem of computing the complexity of inversion \PCId.

%% file: S3ComplexityofInversion.tex
\section{Computation of Polynomial Complexity of Inversion \PCId}
We now consider the problem of computation of the complexity of inversion \PCId. From Definition \ref{DefPCofI} the complexity of inversion of a scalar sequence $S(M)$ is the smallest order $m$ of RRs (\ref{RR}) for which there exists $1\leq d\leq m$ and a polynomial $f$ in $P(m,d)$ such that $f$ is an associated polynomial, the RRs involve a maximal number of LI relations over $S(M)$ and $f$ determines a unique inverse $s_{(-1)}$. The polynomial $f$ is uniquely defined by the co-efficients $\bar{a}$, $\bar{b}$ which arise as the solution of the linear equations (\ref{eqnRR}) involving the Hankel matrix $H(m,d)$. The condition for $m$ to represent the complexity of inversion is captured as in the following theorem. From the equations (\ref{eqnRR}) it follows that the number $n_C$ of columns of $H(m,d)$ is the largest number of monomials possible in a weakly homogeneous $f$ in $P(m,d)$,
\[
n_C=\sum_{i=1}^{d}\binom{m}{i}
\]
The number of equations in (\ref{eqnRR}) is equal to $M-m$ which is the length of the vector $\bar{s}$ which is equal to the number of rows of $H(m,d)$. By definition \ref{DefPCofI}, the complexity $m$ is such that the RRs contain maximal number of LI equations hence in order to ensure that upto the given length $M$ sufficient number of equations are available to solve for the co-efficients of $f$, the number of monomials must satisfy the bound  

\beq\label{Boundonmd}
M-m\geq n_C
\eeq

\begin{theorem}\label{ThPCId}
    A scalar sequence $S(M)$ has \PCId\ of $m$ at a fixed degree $d$ iff $m\geq d$ is the smallest number such that 
    \begin{enumerate}
        \item The bound (\ref{Boundonmd}) holds.
        \item The Hankel matrix $H(m,d)(S(M))$ has maximal rank with respect to $m$ at the sequence considered i.e.
        \beq\label{Maxrank}
        \rank H(m,d)(S(M))=\max_{k\geq d} \rank H(k,d)(S(M))
        \eeq 
        \item There exists a solution to co-efficients $\bar{a}$, $\bar{b}$ of associated polynomials i.e. the condition (\ref{Constraintonh}) holds 
        \end{enumerate}
    The unique inverse for associated polynomials $f$ is 
    \[
    s_{(-1)}=s_{(m-1)}+<v_g(S),\bar{b}>
    \]    
\end{theorem}

\begin{proof}
    First condition ensures that number of equations in RRs are sufficient to compute the largest number of co-efficients possible in $f$. Then the second condition ensures that $m$ is the smallest such that $H(m,d)$ contains maximal number of LI columns when the entire set of RRs for the sequence $S(M)$ upto the length $M$ is taken into account. The third condition is the necessary and sufficient conditions for $f$ to be an associated polynomial and have a unique inverse. Hence $m$ is the complexity of inversion \PCId. The formula for inverse is same as that in the Theorem \ref{TheInverse} for each solution of $\bar{b}$.
\end{proof}

\begin{remark}
     When $m$, $d$ are such that $H(m,d)(S)$ has maximal rank $r(m)$, unless the consistency of equations (\ref{Constraintonh}) holds the associated polynomial does not determine a unique inverse. Hence the complexity of inversion is decided by the solvability of the combined system (\ref{Constraintonh}). When the equations (\ref{Constraintonh}) have a solution the last constraint (\ref{condonh}) may be LI of the equations of RRs (\ref{eqnRR}). Hence the rank of the system can increase to $r(m)+1$ for a scalar sequence.
\end{remark}

\subsection{Inversion of vector valued sequences}
For a vector valued sequence $S(M)$ over $\ftwo^n$, $n>1$, there are $n$ co-ordinate sequences $S(i,M)$ for $i=1,2,\ldots,n$. The inverse of $S(M)$ is then the $n$-tuple $S_{(-1)}$ such that each of its co-ordinates $s_{(-1)}(i)$ is the inverse of each of these co-ordinate sequences defined by a common polynomial $f$ which satisfies the conditions of Theorem \ref{TheInverse} for each $i$. The \PCId\ for a vector sequence follows from extending the conditions of Theorem \ref{ThPCId} for the ensemble of co-ordinate sequences $S(i)$.

\subsubsection{Notations and theorem for inversion of a vector valued sequence}\label{NotationforVectorcase}
The vector case of the inversion problem brings in many complications. Notations for the new matrices involved in the analogous rank conditions as in Theorem \ref{ThInverse} for the scalar sequence will need to be introduced to explain the vector case of inversion. Let $S(M)$ be a vector valued sequence over $\ftwo^n$ and $S(i)$ denote the $i$-th co-ordinate sequence. If $S(M)$ has \PCId equal to $m$, then according to the definition \ref{DefPCofI} an associated polynomial $f$ in $m$ variables of degree $d$ exists such that $m$ is smallest and $f$ determines a unique inverse of each co-ordinate sequence in the equation (\ref{InverseRelationwithcoeff}). We shall continue to denote the vector sequence of length $M$ as $S(M)$ and the co-ordinate sequences as $S(i)$ for $i$-th sequence. Consider following notations (vectors are considered as columns by default except when explicitly stated as row vectors):
\begin{enumerate}
    \item For each $i=1,\ldots,n$, the Hankel matrix $H(m,d)(S(i))$ is denoted $H(i)=[H_1(i),H_2(i)]$
    \item Matrix 
    \[
    \hat{H}(m,d)(S)=\barr{l}H(1)\\\vdots\\H(n)\earr
    \]
    where
    \[
    \hat{H}(1)=
    \barr{l} H_1(1)\\\vdots\\H_1(n)\earr
    \mbox{   }
    \hat{H}_2= 
    \barr{l}H_2(1)\\\vdots\\H_2(n)\earr
    \]
    \item For $S(i)$ the vector $\hat{s}(i)=[s_m(i),\ldots,s_{(M-1)}]^T$ and the vector 
    \[
    \hat{s}=[\hat{s}(1),\ldots,\hat{s}(2)]^T
    \]
    \item $\hat{1}=[1,1,\ldots,1]^T$ an $n$-tuple.
    \item For a function $f$ of $m$ variables with co-efficients $\bar{a},\bar{b}$, $v_h(i)$ denotes the row vector $v_h(S(i))$, $v_g(i)$ denotes the row vector $v_g(S(i))$. Then denote matrices
    \[
    V_h= 
    \barr{l}v_h(1)(S(1))\\\vdots\\v_h(n)(S(n))\earr,
    \mbox{   }
    V_g=
    \barr{l}v_g(1)(S(1))\\\vdots\\v_g(n)(S(n))\earr
    \]
    \item For each $i$, there exists an elementary matrix $E_i$ such that
    \[
    E(i)
    \barr{ll}H_1(i) & H_2(i)\\v_h(i) & 0\earr=
    \barr{ll}H_{11}(i) & H_{12}(i)\\0 & H_{22}(i)\earr
    \]
    where $H_{11}(i)$ is full row rank. Further,
    \[
    E(i)\barr{l}\hat{s}(i)\\1\earr=
    \barr{l}\hat{s}_1(i)\\\hat{s}_2(i)\earr
    \]
    \item There exists an elementary matrix $E$ such that
    \[
    E
    \barr{ll}\hat{H}_1 & \hat{H}_2\\V_h & 0\earr=
    \barr{ll}\hat{H}_{11} & \hat{H}_{12}\\0 & \hat{H}_{22}\earr
    \]
    where $\hat{H}_{11}$ is full row rank and further,
    \[
    E\barr{l}\hat{s}\\\hat{1}\earr=
    \barr{l}\hat{s}_1\\\hat{s}_2\earr
    \]
\end{enumerate}

We can now state the analogue of Theorem \ref{ThInverse} for conditions on invertibility of a vector valued sequence and existence of associated polynomials.

\begin{theorem}[Inversion of a vector sequence]\label{ThVectorInversion}
Let $S(M)$ be a vector sequence and let the number of variables $m$ and largest degree $d$ of polynomials associated with RRs of a vector sequence $S(M)$ be fixed. Then a unique inverse $S_{(-1)}$ exists iff there exist an $f$ in the set of polynomials $F(m,d)$ satisfying following conditions.

There exist associated polynomials $f$ in $P(m,d)$ for the vector sequence $S(M)$, defined by parameters $\bar{a}$, $\bar{b}$ in representation (\ref{PolynforRR}) and determine a unique inverse $S_{-1}$ iff the existence condition
    \beq\label{VecExistencecond}
    \rank \hat{H}_{22}=\rank[\hat{H}_{22},\hat{s}_2]
    \eeq
holds.

All associated polynomials $f$ with the inverse have co-efficients $\bar{a}$, $\bar{b}$ as solutions of
\beq\label{Solnsofcoeff}
\barr{ll}\hat{H}_{11} & \hat{H}_{12}\\0 & \hat{H}_{22}\earr
\barr{l}\bar{a}\\\bar{b}\earr
=
\barr{l}\hat{s}_1\\\hat{s}_2\earr
\eeq
For any one solution $\bar{b}$ of the above equation, the unique inverse for the associated polynomial $f$ with co-efficients $\bar{a}$, $\bar{b}$ is give by
\beq\label{VecInverse}
S_{-1}=S_{(m-1)}+V_g\bar{b}
\eeq
Further, all associated polynomials have a common inverse iff the condition
\beq\label{VecConstancycond}
\rank \hat{H}_{22}(S)=\rank
    \barr{c}\hat{H}_{22}(S)\\V_g(S)\earr
\eeq
holds.
\end{theorem}

\begin{proof}
    The existence condition (\ref{VecExistencecond}) is the necessary and sufficient condition for existence of a polynomial $f$ in $P(m,d)$ to satisfy RRs for all co-ordinate sequences $S(i)$ and existence of inverse for each of these sequences for $i=1,2,\ldots,n$. The inverse is given by the formula (\ref{VecInverse}). Then the uniqueness condition (\ref{VecConstancycond}) is the necessary and sufficient condition for all these polynomials $f$ which satisfy the existence conditions to be the equivalence class corresponding to a unique inverse for each of the co-ordinate sequences $S(i)$.  
\end{proof}

\subsubsection{\PCId\ for a vector sequence}
Computing the complexity \PCId\ of a vector valued sequence $S(M)=\{S(i)(M)\}$ is obtained by extending the Theorem \ref{ThPCId} for a single sequence. Using the notations developed in section \ref{NotationforVectorcase} above for Theorem \ref{ThVectorInversion} it follows that an associated polynomial $f$ for $S(M)$ satisfies the RRs
\[
\hat{H}_1\bar{a}+\hat{H}_2\bar{b}=\hat{s}
\]
and for unique inversion $S(i)(-1)$ to exist for each co-ordinate sequence the equations
\[
V_h\bar{a}=\bar{1}
\]
must be consistent with the RRs. The order $m$ is the \PCId\ for $S(M)$ by definition when the maximum number of LI RRs are captured at $m$ for an appropriate $d$. Moreover the bound between the sequence length $M$ and the order $m$, degree $d$ pair defined by the equation (\ref{Boundonmd}) for each co-ordinate sequence also holds. Since the number of columns $n_C$ of the Hankel matrix $H(S(i))$ does not depend on $i$ or $S(i)$ but only depends on $m,d$, this bound is still the same as in (\ref{Boundonmd}). These considerations prove the following 

\begin{theorem}\label{ThPCIdvector}
    A vector sequence $S(M)=\{S(i)\}$ has \PCId\ of $m$ at a fixed degree $d\leq m$ iff $m$ is the smallest number such that the bound (\ref{Boundonmd}) holds where $n_C$ is the number of columns of the Hankel matrix 
    \[
    \hat{H}(m,d)(S)=\barr{l}H(1)\\\vdots\\H(n)\earr
    \]
    \begin{enumerate}
        \item The matrix $\hat{H}(m,d)(S)$ has maximal rank with respect to $m$ i.e.
        \beq\label{VecMaxrank}
        \rank \hat{H}(m,d)(S)=\max_{k\geq d}\rank \hat{H}(k,d)(S)
        \eeq
        \item The condition (\ref{VecExistencecond}) for existence of associated polynomials $f$ with unique inverse holds.  
        \end{enumerate}
    The unique inverse is determined by associated polynomials $f$ as 
    \[
    S_{(-1)}=S_{(m-1)}+V_g\bar{b}
    \] 
    where $f$ determined by solutions $\bar{a}$, $\bar{b}$ of (\ref{VecExistencecond}).
\end{theorem}

\subsection{Numerical examples to illustrate inversion of vector sequences}
We shall consider three examples below to illustrate the sequence inversion starting from a scalar sequence to vector sequences.

\begin{example}
    Consider the scalar sequence $S=\{0,1,1,1,0,0,0\}$. The MOC is $\geq 3$ since there are two subsequences $\{1,1,1\}$ and $\{1,1,0\}$ with different suffixes. Let $m=3$ be chosen as MOC. The associated polynomial $f(X_0,X_1,X_2)$ is represented in the special form $X_0h(X_1,X_2)+g(X_1,X_2)$ and functions $h$ and $g$ are found from the conditions on these functions to satisfy RRs (\ref{RRincoeff}) and (\ref{InverseRelationwithcoeff}). The RRs are
    \[
    \begin{array}{lclcl}
    0*h(1,1) & + & g(1,1) & = & 1 \\
    1*h(1,1) & + & g(1,1) & = & 0 \\
    1*h(1,0) & + & g(1,0) & = & 0 \\
    1*h(0,0) & + & g(0,0) & = & 0 \\
    \end{array}
    \]
    The inverse relation (\ref{InverseRelationwithcoeff}) is
    \[
    s_{(-1)}h(0,1)+g(0,1)=1
    \]
    From the RRs we get conditions for Truth Tables of $h$ and $g$;
    \[
    \begin{array}{l}
    g(1,1)=1, g(1,0)=D, g(0,0)=D, g(0,1)=D\\
    h(1,1)=1, h(0,1)=1, h(1,0)=g(1,0), h(0,0)=g(0,0).
    \end{array}
    \]
    where $D$ denotes free Boolean value. Let the function $h(X_1,X_2)=a_0+a_1X_1+a_2X_2+a_{12}X_1X_2$. Choosing $g(X_1,X_2)=X_1X_2$ satisfies the truth table conditions. The conditions on co-effcients of $h$ are then
    \[
    a_0+a_1=0,a_0+a_2=1,a_0=0,a_0+a_1+a_2+a_{12}=1
    \]
    These have solution $a_0=a_1=0,a_2=1,a_{12}=0$. Hence $h(X_1,X_2)=X_2$. Hence one associated polynomial for $S$ is
    \[
    f(X_0,X_1,X_2)=X_0X_2+X_1X_2
    \]
    The inverse for this $f$ is $s_{(-1)}=g(0,1)+s_2=0+1=1$
\end{example}

\begin{example}
In this example we have a vector sequence of $2$-tuples. The co-ordinate sequence $S(1)$ is same as in previous example. Hence the RRs and inversion equations are same for truth table values of $h$ and $g$ for $S(1)$. The sequence is 
\[
S=\{
\barr{l}S(1)\\S(2)\earr
\}=
\{
\barr{l}0\\1\earr,\barr{l}1\\1\earr,\barr{l}1\\1\earr,\barr{l}1\\0\earr,
\barr{l}0\\0\earr,\barr{l}0\\0\earr,\barr{l}0\\1\earr
\}
\]
The value of MOC continues to be $m=3$ because $S(2)$ has sub-sequences $\{0,0,0\}$ and $\{0,0,1\}$ with different suffixes. Hence we take the associated polynomial with three variables. However note that a weakly homogenous associated polynomial $f$ cannot satify RRs for $S(2)$ because for such an $f$ the last RR will be $1=f(0,0,0)=0$ which will be contradictory. Hence we consider a non-homogeneous associated polynomial for this vector sequence $S$ in the form
\[
f(X_0,X_1,X_2)=X_0h(X_1,X_2)+g(X_1,X_2)
\]
such that $g(0,0)=1$. The RRs for $S(2)$ are
\[
\begin{array}{lclcl}
1*h(1,1) & + & g(1,1) & = & 0\\
1*h(1,0) & + & g(1,0) & = & 0\\
1*h(0,0) & + & g(0,0) & = & 0\\
0*h(0,0) & + & g(0,0) & = & 1
\end{array}
\]
Hence we get $g(0,0)=1$ and this implies $h(0,0)=1$. Further $h(1,0)=g(1,0)$ and $h(1,1)=g(1,1)$. The inverse relation gives the constraint 
\[
S(2)_{(-1)}h(1,1)+g(1,1)=1
\]
Since $h(1,1)=1$ for uniqueness of inverse, $S(2)_{(-1)}=g(1,1)+1$. From RRs of $S(1)$ we have $g(1,1)=1$, hence $S(2)_{(-1)}=0$. The constraint $h(1,0)=g(1,0)$ appears in both RRs hence $g(1,0)$ is free. Similarly $g(0,1)$ is free. Hence the truth tables for $h$ and $g$ for satisfying RRs of the vector sequence and inversion condition are
\[
\begin{array}{lll}
(X_1,X_2) & h(X_1,X_2) & g(X_1,X_2)\\
\hline
(0,0) & 1 & 1\\
(1,0) & g(1,0) & D\\
(0,1) & 1 & D\\
(1,1) & 1 & 1
\end{array}
\]
Condsider the forms $h=a_0+a_1X_1+a_2X_2+a_{12}X_1X_2$ and $g=b_0+b_1X_1+b_2X_2+b_{12}X_1,X_2$. Then from the truth tables $a_0=1=b_0$. This implies $a_1+a_2+a_{12}=0$, $b_1+b_2+b_{12}=0$. From condition at $(1,0)$ and $h(1,1)+g(1,1)=0$ it follows that $a_1+b_1=0$, $a_2+a_{12}=0$, $b_2+b_{12}=0$. Hence the only choices possible are $a_2=b_2=1$ or $a_2=b_2=0$, $a_1=b_1=0$, $a_{12}=b_{12}=1$ or $a_{12}=b_{12}=0$. These give two non-homogenous associated polynomials
\[
\begin{array}{lcll}
f_1 & = & X_0+1 & \mbox{for }h(X_1,X_2)=g(X_1,X_2)=1\\
f_2 & = & X_0(1+X_2+X_1X_2)+(1+X_2+X_1X_2) & \mbox{for other parameters}
\end{array}
\]
The inverse of $S$ for $f_1=X_0+1$ is
\[
\barr{l}S(1)_{(-1)}\\S(2)_{(-1)}\earr=\barr{l}1+1=0\\1+1=0\earr
\]
The inverse for $f_2=X_0(1+X_2+X_1X_2)+(1+X_2+X_1X_2)$ is also
\[
\barr{l}S(1)_{(-1)}\\S(2)_{(-1)}\earr=\barr{l}1+1=0\\1+1=0\earr
\]
This happens because the terms $s_1$, $s_2$ in both co-ordinate sequences are the same.
\end{example}

%% file: S4Condnsforspecialpoly.tex
\section{Complexity of inversion relative to a monomial set and associated polynomials with special structure}
The notion of inverse of a sequence arises primarily from periodic sequences where the prefix (i.e. the inverse) is the last element of the sequence. Such an inverse is defined by the periodicity RR $s_{(N+k)}=s_{k}$ for $k=0,1,2,\ldots,(N-1)$ which is of order $m=N$, degree $d=1$ and the associated polynomial $f(X_0,\ldots,X_{(N-1)})=X_{(N-1)}-X_{0}$. In this paper we have extended the notion of inverse with respect to associated polynomials of higher degrees $d>1$ in the class of polynomials $P(m,d)$ where $1\leq d\leq m$. Under the assumption that the sequence $S(M)$ of length $M$ is a partial sequence of a periodic sequence of (some unknown) length $N$, the periodicity RR gives an associated polynomial with order larger than $M$, hence existence of associated polynomials is guaranteed due to the periodicity assumption. 

In previous section conditions for polynomials in $P(m,d)$ to satisfy RRs (or serve as associated polynomials) and determine an inverse are established. Further, conditions to determine the complexity of inversion \PCId\ of sequences are developed for general polynomials in $P(m,d)$. However $P(m,d)$ contains many polynomials with special structures such as linear, homogeneous, symmetric with permutation symmetries for restricted variables even for same $m$ and $d$.
These polynomials with special structure have monomials which belong to a special class among all monomials in $P(m,d)$ and are chosen a-priori. Hence it is appropriate to consider distinct Hankel matrices when these monomials are evaluated over the sequence. Let $\cal{M}$ denote a subset of monomials in $P(m,d)$. Considering an ordering $x_0<x_1<x_2\ldots x_{(n-1)}$ of variables, the set of monomials $m(i_1,i_2,\ldots,i_k)=x_{i_1}\ldots x_{i_k}$ with the convention $i_1<\ldots<i_k$ in $P(m,d)$ is also ordered as  
\[
m(i_1,i_2,\ldots,i_k)\leq m(j_1,j_2,\ldots,j_l)
\]
when $k\leq l$ if $i_1\leq j_1,\ldots,i_k\leq j_k$. Otherwise 
\[
m(i_1,i_2,\ldots,i_k)> m(j_1,j_2,\ldots,j_l)
\]
if $k>l$.

The problem addressed in this section is, given a subset of monomials $\cal{M}$ in $P(m,d)$ what are the associated polynomials $f$ defined by linear forms over $\cal{M}$, when do they determine an inverse and what is an appropriate notion of the complexity of inversion. Since we hypothesise the complexity to depend on $\cal{M}$ chosen, we shall denote this as $C(M)(S)$ for a sequence $S$. Thus in this section the conditions for existence of inverse are determined with respect to associated polynomials which are linear forms of monomials in $M$.

\subsection{Hankel matrix of $\cal{M}$}
For a monomial $m(i_1,i_2,\ldots,i_k)$, its evaluation over a scalar sequence $S(M)$ is defined as the sequence of values
\[
\{m(i_1,i_2,\ldots,i_k)(\sigma^jS)\}=\{m(s_{i_1+j}s_{i_2+j}\ldots s_{i_k+j}),j=0,1,2,\ldots\}
\]
Considering the representation of the associated polynomials as in (\ref{PolynforRR}) and using the monomial representation of polynomials $h$ and $g$ in this representation, the co-efficients of a general $f$ which is a linear form over $\cal{M}$ are classified in two types denoted by vectors $\bar{a}$ (co-effcients of $f$) and $\bar{b}$ (co-effcients of $g$). There are then two types of monomials in $\cal{M}$, one which are multiples of $X_0=m(0)$ and others which do not have $X_0$ as a factor. Hence any $f$ with monomials in the set ${\cal M}$ is of the form
\beq\label{PolyfunM}
f=m(0)(a_0+\sum_{i}a_im_i)+(\sum_{j}b_jm_j)
\eeq
such that monomials $m(0)=X_0$, $m(0)m_i$ and $m_j$ belong to $\cal{M}$.
\[
\begin{array}{ll}
f(X_0,X_1,\ldots,X_{(m-1)})= & m(0)[a_0+\sum_{1:(m-1)}a_{(1:(m-1))}m_{(1:(m-1)}]+\\
& [\sum_{1:(m-1)}b_{(1:(m-1))}m_{(1:(m-1))}]
\end{array}
\]
where $1:(m-1)$ denotes the indices $1\leq i_1<i_2\ldots<i_{r}\leq (m-1)$ and $a_{(1:(m-1))}$ is the co-efficient of the monomial $m_{(1:(m-1))}$ in the summation expression. 

Given such set $\cal{M}$ of monomials we define a Hankel matrix of evaluation of $\cal{M}$ at $S$ to be the matrix
\beq\label{HankelmatrixofM}
\begin{array}{lcl}
H({\cal M})(S) & = &
\barr{llllll}
s_0a_0 & \ldots & s_0m_i(S) & \ldots & m_j(S) & \ldots\\
s_1a_0 & \ldots & s_1m_i(\sigma S) & \ldots & m_j(\sigma S) & \ldots\\
\vdots & \vdots & \vdots & \vdots & \vdots & \vdots\\
s_{(M-m-1)}a_0 & \ldots & s_{(M-m-1)}m_i(\sigma^{(M-m-1)}S) & \ldots & m_j(\sigma^{(M-m-1)}S) & \ldots
\earr
\\
 & = & [H_1(S),H_2(S)]
\end{array}
\eeq
Where $H_1(S)$ is the submatrix of monomial evaluations of the type $m(0)m_i(\sigma^jS)$ and $H_2(\sigma^jS)$ that of evaluations $m_j(\sigma^kS)$.
An associated polynomial $f$ which is a linear form over monomials in ${\cal M}$ can be expressed in the form (\ref{PolyfunM}) with two sets of coefficients $\bar{a}$ and $\bar{b}$. The columns of $H({\cal M})(S)$ are ordered in the order of the monomial terms in $f$ in the form (\ref{PolyfunM}) as in the expression $[H_1(S),H_2(S)]$ of $H({\cal M})(S)$. Further the evaluation of monomial terms $m_i$ and $m_j$ in the expression of $f$ are grouped in two vectors denoted as
\[
\begin{array}{lcl}
\bar{m}_1(S) & = & [m_i(S)]\\
\bar{m}_2(S) & = & [m_j(S)]
\end{array}
\]
If $S(M)=\{S_i(M)\},i=1,2,\ldots,n\}$ is a vector sequence where each co-ordinate sequence is denoted as
\[
S_i(M)=\{s^i_0,s^i_1,\ldots,s^i_{(M-1)}\}
\]
then the evaluation of a monomial $m(i_1,i_2,\ldots,i_k)$ over $S$ is the vector sequence
\[
\{\bar{m}(i_1,i_2,\ldots,i_k)(j)\}=\{\bar{m}(s^i_{i_1+j}s^i_{i_2+j}\ldots s^i_{i_k+j}),j=0,1,2,\ldots,(M-1-i_k)\}
\]
for $i=1,2,\ldots,n$. Then the Hankel matrix of ${\cal M}$ relative to $S(M)$ is the matrix of vector valued evaluations of $m_i$ and $m_j$ over the vector sequence $S(M)$ and that of $m(0)$ as the vector sequence $s^i_k, k=0,1,\ldots,(M-m-1)$. The Hankel matrix of ${\cal M}$ is still denoted as $H({\cal M})(S)$ while the evaluations of monomials $m_i$ and $m_j$ in expression of evaluation of $f$ over $S(M)$ are denoted as
\[
\begin{array}{lcl}
M_1(S) & = & [m_i(S)]\\M_2(S) & = & [m_j(S)]
\end{array}
\]

\subsubsection{Illustrative example for a Hankel matrix}
Let ${\cal M}=\{m_0,m_1,m_2,m_3\}=\{X_0X_1,X_0X_2,X_1X_2,X_2X_3\}$ and a sequence is given a 
\[S=\{s_0,s_1,s_2,s_3,s_4,s_5,s_6,s_7\}
\]
Then the order $m=4$ and $M=8$. Hence maximum shift of $S$ is possible at the power $M-m-1$ by operator $\sigma^{3}$. Then the evaluations of monomials is given by
\[
\begin{array}{llllll}
m_0(S) & = & s_0s_1 & m_0(\sigma S) & = & s_1s_2\\
m_1(S) & = & s_0s_2 & m_1(\sigma S) & = & s_1s_3\\
m_2(S) & = & s_1s_2 & m_2(\sigma S) & = & s_2s_3\\
m_3(S) & = & s_2s_3 & m_3(\sigma S) & = & s_3s_4\\
m_0(\sigma^2 S) & = & s_2s_3 & m_0(\sigma^3 S) & = & s_3s_4\\
m_1(\sigma^2 S) & = & s_2s_4 & m_1(\sigma^3 S) & = & s_3s_5\\
m_2(\sigma^2 S) & = & s_3s_4 & m_2(\sigma^3 S) & = & s_4s_5\\
m_3(\sigma^2 S) & = & s_4s_5 & m_3(\sigma^3 S) & = & s_5s_6
\end{array}
\]
The Hankel matrix of ${\cal M}$ over $S$ is then given by
\[
\begin{array}{ll}
H({\cal M})(S)=[H_1(S)|H_2(S)]=&
\barr{ll|ll}
m_0(S) & m_1(S) & m_2(S) & m_3(S)\\
m_0(\sigma S) & m_1(\sigma S) & m_2(\sigma S) & m_3(\sigma S)\\
m_0(\sigma^2 S) & m_1(\sigma^2 S) & m_2(\sigma^2 S) & m_3(\sigma^2 S)\\
m_0(\sigma^3 S) & m_1(\sigma^3 S) & m_2(\sigma^3 S) & m_3(\sigma^3 S)
\earr
\end{array}
\]
Hence it follows that unless the sequence length $M$ is sufficiently long, the number of column $n_C$ of the Hankel matrix may not be smaller than the number of rows $M-m$. In realistic situations the sequence length is always sufficiently large so that we have more equations in RRs than the number of monomials. 

\subsection{Complexity of Inversion relative to ${\cal M}$}
First we need to define the notion of complexity of inversion of a sequence $S$ using RRs defined by an associated polynomial $f$ which is a linear form over an a-priori fixed set of monomials ${\cal M}$. Then we state the conditions on existence of such an associated polynomial defining an inverse for the sequence $S$. Complexity of inversion \PCId as defined earlier, is the order of the polynomials satisfying the RRs for the sequence upto the given length which reflect maximal number of linearly independent equations satisfied by the co-efficients of the associated polynomials (also the rank of the Hankel matrix $H(m,d)$) and also satisfying the inversion relation. 

Let a subset of monomials ${\cal M}$ is chosen from the set of all monomials in $P(m,d)$ and $m$ is the maximal number of variables in the monomials of ${\cal M}$. The set ${\cal M}$ is fixed for solving the inversion problem. Let $H({\cal M}(S))$ be the Hankel matrix of evaluation of ${\cal M}$ at $S$ defined in (\ref{HankelmatrixofM}) upto the sequence length $M$. Since the number of columns $n_C$ of $H({\cal M})(S)$ is the largest number of co-efficients in the associated polynomials possible, the number of rows of the matrix which is same as the number of RRs is at least as much as $n_C$. Hence we have the bound (\ref{Boundonmd})
\[
M-m\geq n_C
\]
This bound makes the sequence length sufficiently long to capture all the LI RRs for the sequence $S$. We shall hereafter always assume that this bound is satisfied for any choice of a monomial set ${\cal M}$ used for inversion.

\begin{definition}[Complexity of Inversion relative to ${\cal M}$]
Let $m$ be the largest number of variables in any monomials in ${\cal M}$ (called as order of ${\cal M}$ and $r$ be the rank of the Hankel matrix $H({\cal M}(S))$ for the sequence $S(M)$ given upto the length $M$ and the system of equations
\[
H({\cal M}(S))
\barr{c}\bar{a}\\\bar{b}\earr
=
\barr{c}\bar{s}\\1\earr
\]   
has a solution. Then $C({\cal M})=\sqrt{rm}$ is defined as the Complexity of Inversion of $S$  relative to ${\cal M}$. If this system does not have a solution then $C({\cal M})$ is undefined. 

For a vector sequence $S(M)=\{S_i(M)\},i=1,2,\ldots,n\}$, 
\end{definition}

The rationale behind the above definition is that when the system of equations (\ref{RRincoeff}) and condition for inversion (\ref{InverseRelationwithcoeff}) are both satisfied by an associated polynomial which is a linear form over $\cal{M}$, the complexity of solving for the inverse is governed by both factors the number of variables $m$ and the maximal rank of the Hankel matrix. 

\begin{theorem}\label{Th:InversionoverM}
Let $S(M)$ be a scalar sequence of length $M$ and ${\cal M}$ be a set of monomials which has finite complexity relative to $S(M)$. Let $m$ be the maximal number of variables in monomials of $\cal{M}$ and $r$ is the rank of $H(\cal{M})(S)$ upto the sequence length $M$. Then there exists a subset ${\cal M}_1\subset\cal{M}$ (whose Hankel matrix is denoted as)
\[
H({\cal M}_1)(S)=[H_{11}(S),H_{12}(S)]
\]
such that the equations
\[
\begin{array}{lcl}
H_{11}(S)\bar{a}+H_{12}(S)\bar{b} & = & \bar{s}\\
<\bar{m}_1(S),\bar{a}> & = & 1
\end{array}
\]
have a unique solution $\bar{a}$, $\bar{b}$ for the co-efficients of an associated polynomial $f$ which is a linear form over ${\cal M}_1$. The inverse of $S$ is given by
\[
s_{(-1)}=s_{(m-1)}+<\bar{m}_2(S),\bar{b}>
\]
\end{theorem}

Analogous theorem for inverse of a vector sequence is as follows.

\begin{theorem}\label{Th:InversionoverMvectorseq}
\end{theorem}

\subsection{Inverse using Linear Complexity $C(m,1,L)(S)$}
Let $L(X_0,\ldots,X_{(m-1)})$ denote a linear form $\sum_{i=0}^{(m-1)}a_iX_i$ in $m$ variables. Hence the associated polynomial is $L$ and has fixed degree $d=1$. $L$ must satisfy RRs (\ref{RR}). The first condition on bounds between the length of $S$ and complexity is (\ref{Boundonmd}) is $M\geq 2m$. Assuming for simplicity that $S$ is a scalar sequence, the Hankel matrix of $f$ at $S$ is 
\beq\label{LCHankel}
H(m,1,L)(S)=[H_1|H_2]
\barr{l|llll}
s_0 & s_1 & s_2 & \ldots & s_{(m-1)}\\
s_1 & s_2 & s_3 & \ldots & s_m\\
\vdots & \vdots & \vdots & \ldots & \vdots\\
s_{(M+m-1)} & s_{(M-m)} & s_{(M-m+1)} & \ldots & s_{(M+2m-2)}
\earr
\eeq
Since $d=1$ is fixed let this matrix be denoted as $H(m)$, the first condition of Theorem \ref{ThPCId} is that 
\[
\rank H(S)=m
\]
is maximal with respect to the sequence length $M$. Since $\rank H(S)=m$ equals the number of columns in $H(S)$ there is a unique solution to co-efficients of the linear function $L$. For the linear $f$ the polynomial $h(X_1,\ldots,X_{(m-1)})=a_0$. Hence the second condition of Theorem \ref{ThPCId} fixes the value of $a_0=1$. Other co-effcients $a_1,\ldots,a_{(m-1)}$ are uniquely determined. The function 
\[
g(X_1,\ldots,X_{(m-1)}=\sum_{i=1}^{(m-1)}a_iX_i
\]
hence $v_g(S)=[s_1,s_2,\ldots,s_{(m-1)}]$. The formula for inverse in Theorem \ref{ThPCId} gives
\beq\label{LCinverse}
s_{(-1)}=s_{(m-1)}+\sum_{i=1}^{(m-1)}a_is_i
\eeq

This discussion proves the 

\begin{theorem}\label{ThLC}
For degree $d=1$, an inverse exists for a linear associated polynomial iff there is $m$ the smallest number such that $M\geq 2m$, the matrix
\[
H(m,1,L)(S)=[H_1|H_2]
\]
has maximal column rank which is $m$ and the equations (\ref{Constraintonh}) have a solution. The inverse exists uniquely and is given by (\ref{LCinverse}) while the associated polynomial is
\[
f(X_0,X_1,\ldots,X_{(m-1)}=X_0+\sum a_iX_i
\] 
where the co-efficients $\bar{b}=[a_1,\ldots,a_{(m-1)}]^T$ are solved from the equation
\[
H_1+H_2\bar{b}=\bar{s}
\]
the complexity of inversion $C(m,1,L)(S)=m$.
\end{theorem}

This is precisely the inverse predicted from the minimal polynomial of the sequence from shortest linear RRs referred in previous work \cite{Sule2,Sule3,Sule4}. The condition $a_0=1$ confirms with the condition that the constant term of the minimal polynomial (which is precisely $a_0$) is non-zero. 

\subsection{Inverse using general second degree polynomials}
First consider the general (assumed weakly homogeneous) second degree polynomial $P(m,2)$ in $m$ variables
\[
f=\sum_{i=0}^{(m-1)}a_iX_i+\sum_{i=0}^{(m-2)}\sum_{j>i}^{(m-1)}a_{ij}X_iX_j
\]
This $f$ is expressed in the two term form
\[
f=X_0h(X_1,\ldots,x_{(m-1)},\bar{a})+g(X_1,\ldots,X_{(m-1)},\bar{b})
\]
with co-efficients $\bar{a}$, $\bar{b}$ of $h$ and $g$. The functions $h$, $g$ have the form
\[
\begin{array}{lcl}
h & = & a_{00}+\sum_{j=1}^{(m-1)}a_{0j}X_j\\
g & = & \sum_{i=1}^{(m-1)}b_iX_i+
\sum_{i=1}^{m-2}\sum_{j>i}^{(m-1)}b_{ij}X_iX_j
\end{array}
\]
Hence
\[
\begin{array}{lcl}
\bar{a} & = & [a_{00},a_{01},\ldots,a_{0(m-1)}]\\
\bar{b} & = & [b_1,\ldots,b_{(m-1)},b_{12},\ldots,b_{(m-2)(m-1)}]
\end{array}
\]
Then we have the Hankel matrix of $f$ in definition \ref{HankelMatrix} written according to the presentation of co-efficients of $f$ as in (\ref{PolynforRR}). For example the Hankel matrix of an $f$ in $P(3,2)$ is of the form, for a sequence of length $M$,
\[
H(S(M))=[H_1|H_2]=
\barr{lll|lll}
s_0 & s_0s_1 & s_0s_2 & s_1 & s_2 & s_1s_2\\
s_1 & s_1s_2 & s_1s_3 & s_2 & s_3 & s_2s_3\\
\vdots & \vdots & \vdots & \vdots & \vdots & \vdots\\
s_{(M-4)} & s_{(M-4)}s_{(M-3)} & s_{(M-4)}s_{(M-2)}
& s_{(M-3)} & s_{(M-2)} & s_{(M-3)}s_{(M-2)}
\earr
\]
Vectors $v_h$, $v_g$ for this polynomial $f$ in $P(3,2)$ defined by polynomials $h$, $g$ evaluated at $S$ are
\[
\begin{array}{lcl}
v_h(S) & = & [1,s_1,s_2]\\
v_g(S) & = & [s_1,s_2,s_1s_2]
\end{array}
\]

For the general weakly homogenous polynomial $f$ in $P(m,2)$ the various numbers can be counted as follows.
\begin{enumerate}
    \item Number of co-effcients in $f$: $N_c=m+{m\choose 2}=(1/2)m(m+1)$.
    \item Number of co-efficients in $\bar{a}$: $N(\bar{a})=m$. Equal to the length of $v_h(S)$.
    \item Number of co-efficients in $\bar{b}$: $N(\bar{b})=(m-1)+{(m-1)\choose 2}-(m-1)={m\choose 2}$. Equal to the length of $v_g(S)$.
    \item Number of columns in $H(S(M))$: $n_C=m+{m\choose 2}=N_c$.
    \item Number of RRs of order $m$ in the sequence length $M$: $n_R=M-m$. These are equal to the number of rows in the Hankel matrix of $f$ evaluated at $S$.
\end{enumerate}
The Hankel matrix of $f$ is the matrix
\[
H(S(M))=[H_1(S)|H_2(S)]
\]
where $H_1$ has $m$ columns and $H_2$ has ${m\choose 2}$
columns.

\begin{theorem}\label{ThGendeg2poly}
    For $d=2$ and $f$ a general homogeneous polynomial of degree $2$, the complexity of inversion \PCId is equal to the smallest $m$ such that 
    \[
   \rank H(S(M))=\rank [H_1(S)|H_2(S)]
    \]
    is maximal with respect to $m$ and the equations (\ref{Constraintonh})
    \[
\begin{array}{rcl}
H_1(S)\bar{a}+H_2(S)\bar{b} & = & \bar{s}\\
v_h(S)\bar{a} & = & 1
\end{array}
\]
have a solution. The unique inverse is given by
    \[
    s_{(-1)}=s_{(m-1)}+<v_g(S),\bar{b}>
    \]
\end{theorem}

\begin{proof}
    From the definition of complexity of inversion, the Hankel matrix of $f$ must have the maximal rank. Then the existence and unique solvability of the inverse is equivalent to the existence of solutions to equations (\ref{Constraintonh}). The inversion relation (\ref{InverseRelationwithcoeff}) becomes $s_{(m-1)}=s_{(-1)}h(S)+g(S)$ which gives the inverse as stated since $h(S)=<v_h(S),\bar{a}>=1$ and $g(S)=<v_g(S),\bar{b}>$.
\end{proof}

\subsection{Inverse using homogeneous second degree polynomials}
In this case we consider $f$ in $P(m,2)$ which is homogeneous
\[
f=\sum_{i=0}^{(m-2)}\sum_{j>i}^{(m-1)}a_{ij}X_iX_j
\]
This polynomial has representation in the form (\ref{PolynforRR}) as
\[
\begin{array}{lcl}
f & = & X_0h(X_1,\ldots,X_{(m-1)})+g(X_1,\ldots,X_{(m-1)})\\
 & & X_0(\sum_{i=1}^{(m-1)}a_{0i}X_i)+
(\sum_{i=1}^{(m-2)}\sum_{j>i}^{(m-1)}a_{ij}X_iX_j)
\end{array}
\]
For example for $m=4$, $d=2$, the Hankel matrix (\ref{HankelMatrix}) of such an $f$ is
\[
\begin{array}{lcl}
H(S) & = & [H_1(S)|H_2(S)]\\
 & = &
\barr{lll|lll}
s_0s_1 & s_0s_2 & s_0s_3 & s_1s_2 & s_1s_3 & s_2s_3\\
s_1s_2 & s_1s_3 & s_1s_4 & s_2s_3 & s_2s_4 & s_3s_4\\
\vdots & \vdots & \vdots\\
s_{(M-5)}s_{(M-4)} & s_{(M-5)}s_{(M-3)} & s_{(M-5)}s_{(M-2)} & s_{(M-4)}s_{(M-3)} & s_{(M-4)}s_{(M-2)} & s_{(M-3)}s_{(M-2)}
\earr
\end{array}
\]

\[
\begin{array}{lll}
 v_h(S) & = & [s_0s_1,s_0s_2,s_0s_3] \\
 v_g(S) & = & [s_1s_2,s_1s_3,s_2s_3]
\end{array}
\]
Various numbers associated with the polynomial $f$ are as follows:
\begin{enumerate}
    \item Number of co-effcients in $f$: $N_c={m\choose 2}=(1/2)m(m-1)$.
    \item Number of co-efficients in $\bar{a}$: $N(\bar{a})=(m-1)$. Equal to the length of $v_h(S)$.
    \item Number of co-efficients in $\bar{b}$: $N(\bar{b})={(m-1)\choose 2}$. Equal to the length of $v_g(S)$.
    \item Number of columns in $H(S(M))$: $n_C=(m-1)+{(m-1)\choose 2}=N_c$.
    \item Number of RRs of order $m$ in the sequence length $M$: $n_R=M-m$. These are equal to the number of rows in the Hankel matrix of $f$ evaluated at $S$.
\end{enumerate}

\begin{theorem}\label{ThHomf}
If the associated polynomial $f$ is restricted to be homogeneous of degree $2$, then the \PCId\ $m$ is the smallest number such that the Hankel matrix $H(S))$ has maximal rank with respect to $m$ and the system of equations
\[
\begin{array}{rcl}
H_1(S)\bar{a}+H_2(S)\bar{b} & = & \bar{s}\\
v_h(S)\bar{a} & = & 1
\end{array}
\]
has a solution $\bar{a}$, $\bar{b}$. The unique inverse is given by
\[
s_{(-1)}=s_{(m-1)}+<v_g(S),\bar{b}>
\]
\end{theorem}

The proof follows on similar lines as that of proof of Theorem \ref{ThGendeg2poly}.

\subsection{Number of solutions of associated polynomials}
Theorems in the previous section gave necessary and sufficient conditions for existence of associated polynomials of smallest order $m$ to achieve a unique inverse of a given sequence and utilise the maximal number of LI RRs for a given length of the sequence. In this section the definition of complexity of inversion of a sequence is extended over a fixed set of monomials ${\cal M}$. Further, conditions for existence of inverse are extended over ${\cal M}$ and a unique solution of the associated polynomial as a linear form over ${\cal M}$ are developed. Analogous results are stated for polynomials of special structures to serve as associated polynomials and achieve inverse while utilising maximal number of LI RRs. 

We now determine the number of associated polynomials which have give a unique inverse for the sequence. Given a sequence $S(M)$ of length $M$, if $f$ is a polynomial in $P(m,d)$ with a fixed structure which fixes the types of monomials present in $f$, let $m$ be the smallest at which $H(m,d,f)(S)=H(S)$ achieves its maximal rank $r(m)$ with respect to $m$ for some $d$. Then there exist $r(m)$ LI columns in $H(S)$. Moreover by Theorem \ref{TheInverse}, an asscociated polynomial $f$ exists which has a unique inverse iff the linear system
\[
\barr{ll}
H_1 & H_2\\v_h & 0
\earr
\barr{l}
\bar{a}\\\bar{b}
\earr
=
\barr{l}
\bar{s}\\1
\earr
\]
has a solution. Since the last equation $<v_h,\bar{a}>=1$ can at most increase the rank of this system to $(r(m)+1)$, there exist at most $(r(m)+1)$ LI columns of the matrix
\[
H_h=
\barr{ll}
H_1 & H_2\\v_h & 0
\earr
\]
which form a submatrix $H_{\mbox{max}}$ and a matrix $A_h$ such that all solutions to $\bar{a}$, $\bar{b}$ can be determined in terms of the following linear system.

It follows that the number of associated polynomials has the bound
\[
n_f\leq 2^{n_C-r(m)}
\]
where $n_C$ is the number of columns of $H(m,d,f)$ also equal to the maximum number of monomials in $f$. For each of these associated polynomials we have a unique inverse. From this discussion we have 

\begin{theorem}\label{ThUniqueAssPoly}
If $N_c$ is the number of monomials for some $m,d$ where $m$ is the \PCId\ of a sequence $S(M)$ and $r(m)$ is the rank of $H(m,d)(S)$ then there exist at most $2^{N_c-r(m)}$ distinct associated polynomials which determine inverses of $S(M)$. If $f$ is chosen with an a-priori structure (or monomials) with order $m$ and degree $d$, has $N_c$ number of monomials and has inverse, then there is unique such $f$ iff $N_c=r(m)$ where $r(m)$ is the \PCId\ of $S(M)$.  
\end{theorem}

\begin{proof}
    When $m$ is \PCId\ then rank of $H(m,d)=r(m)$ is maximal and the equations (\ref{Constraintonh}) are consistent. Then the dimension $d_C$ of the kernel of the linear system (\ref{Constraintonh}) is $(n_C-r(m)-1)\leq d_c\leq (n_C-r(m))$ where $n_C$ is the number of columns in the matrix of these equations which is same as $N_c$ the number of monomials. Hence there are at most $2^{N_c-r(m)}$ solutions to associated polynomials. Since the field is binary the number of possible solutions is bounded by $2^{(N_c-r(m)}$. Hence there is a unique solution to the set of associated polynomials when $d_c=0$ which is same as $N_c=r(m)=\rank H(m,d,f)(S)$ for any polynomials $f$ with a fixed structure (or fixed set of monomials).
\end{proof}

\subsubsection{Number of solutions of $f$ relative to a fixed set of monomials}
For a fixed set of monomials ${\cal M}$ if the order is $m$ and the Hankel matrix $H({\cal M})(S)$ has rank $r$ then for every subset ${\cal M}_1\subset{\cal M}$ which has $\rank H({\cal M}_1)(S)=r$ there is a unique associated polynomial defined by solutions of co-efficients in the equation
\[
H_{11}(S)\bar{a}+H_{12}(S)\bar{b}=\bar{s}
\]
And a unique universe as shown in Theorem \ref{Th:InversionoverM}. Hence the rank of the system combined system of linear equations is $\leq(r+1)$ and thus the number of associated polynomials $n_f$ which determine an inverse is bound by
\[
2^{|{\cal M}|-(r+1)}\leq n_f\leq 2^{|{\cal M}|-r}
\]

\subsection{FSR polynomials satisfying the Golomb's condition for non-singularity}
We now investigate existence of associated polynomials $f$ for a vector sequence $S(M)$ which model all the co-ordinate sequences as sequences generated by FSRs whose feedback polynomials are $f$ and determine a unique inverse for each co-ordinate sequence. This implies that the FSR represented by the associated polynomial is non-singular. The Golomb's condition for non-singularity of an FSR with feedback function $f(X_0,X_1,\ldots,X_{(m-1)})$ is that
\beq\label{Golombcond}
f=X_0+g(X_1,\ldots,X_{(m-1)},\bar{b})
\eeq
where $\bar{b}$ is the vector of co-efficients of the polynomial $g$. The polynomial $g$ is assumed to be weakly homogeneous. The sequence $S(M)$ is then invertible since the condition for invertibility (\ref{InverseRelationwithcoeff}) requires that for each of the co-ordinate sequences $S(i)$
\beq\label{GolombInverse}
S(i)_{(m-1)}=S(i)_{(-1)}+g(S(i)_0,S(i)_1,\ldots,S(i)_{(m-2)}
\eeq
Hence $S(i)_{(-1)}$ can always be calculated for any associated polynomial $f$ satisfying the Golomb's condition (\ref{Golombcond}) and has the value
\beq\label{GolombInverse2}
S(i)_{(-1)}=S(i)_{(m-1)}+<v_g(S),\bar{b}>
\eeq
The polynomial of the form (\ref{Golombcond}) is already in the form (\ref{PolynforRR}) used for defining the Hankel matrix and expressing the condition for $f$ to be an associated polynomial. Hence we can state

\begin{theorem}\label{NonsingularFSR}
    A polynomial $f$ satisfying the Golomb's condition (\ref{Golombcond}) is an associated polynomial for a scalar sequence $S(M)$ iff
    \[
    [s_m+s_0,s_{(m+1)}+s_1,\ldots,s_{(M-1)}+s_{(M-m-1)}]^T
    \]
    is in the span of columns of the matrix
    \[
    H(S)=
    \barr{l}
    v_g(S)\\v_g(\sigma S)^T\\\vdots\\v_g(\sigma^{M-m-1} S)
    \earr
    \]
    Every such polynomial defines a unique inverse as in (\ref{GolombInverse2}).
\end{theorem}

Since the matrix $H$ may not have maximal rank equal to the number of monomials in the polynomial $g$, a non-singular FSR which determines an inverse for a given sequence may not necessarily be unique. 

\subsubsection{Illustration of a non-singular FSR inverting a sequence}
Consider an order $3$ feedback polynomial of a non-singular FSR
\[
f(X_0,X_1,X_2)=X_0+g(X_1,X_2)
\]
where $g(X_1,X_2)=X_1+X_2+X_1X_2$. Next consider a sequence $\{s_0,s_1,\ldots,s_7\}$ of length $8$. Then the RRs defined by $f$ for $S(8)$ are
\[
\begin{array}{lcl}
s_3 & = & s_0+s_1+s_2+s_1s_2\\
s_4 & = & s_1+s_2+s_3+s_2s_3\\
s_5 & = & s_2+s_3+s_4+s_3s_4\\
s_6 & = & s_3+s_4+s_5+s_4s_5\\
s_7 & = & s_4+s_5+s_6+s_5s_6
\end{array}
\]
Hence the polynomial $f$ is an associated polynomial of $S(8)$ iff above RRs hold. The inverse is given by
\[
s_{(-1)}=s_2+g(s_0,s_1)=s_0+s_1+s_2+s_0s_1
\]
The vector 
\[
v_g(S)=[s_1,s_2,s_1s_2]
\]
The matrix $H(S)$ is 
\[
\barr{l}
v_g(S)\\v_g(\sigma S)\\v_g(\sigma^2 S)\\v_g(\sigma^3 S)\\v_g(\sigma^4 S)
\earr
=
\barr{llll}
s_1 & s_2 & s_1s_2\\s_2 & s_3 & s_2s_3\\s_3 & s_4 & s_3s_4\\s_4 & s_5 & s_4s_5\\
s_5 & s_6 & s_5s_6
\earr
\]
Conversely if $g=b_0X_1+b_1X_2+b_2X_1X_2$ is a polynomial with $b_i$ in $\ftwo$, then there exists a polynomial of the form $f=X_0+g(X_1,X_2)$ which is an associated polynomial of $S$ iff
\[
H(f,S)
\barr{l}1\\\bar{b}\earr=
\barr{llll}
s_0 & s_1 & s_2 & s_1s_2\\
s_1 & s_2 & s_3 & s_2s_3\\
s_2 & s_3 & s_4 & s_3s_4\\
s_3 & s_4 & s_5 & s_4s_5\\
s_4 & s_5 & s_6 & s_5s_6
\earr
\barr{l}
1\\b_0\\b_1\\b_2
\earr
=
\barr{l}
s_3\\s_4\\s_5\\s_6\\s_7
\earr
\]
has a solution $\bar{b}$. 

%% file: S5ExpectedComplexity.tex
\section{Expected Complexity of Inversion of Sequences and Local Inversion of maps}
In this brief section we address the issue of expected complexity of inversion. The problem of inversion in practice as well as that of local inversion of a map always arises with respect to a partially available data as described in the section \ref{PartialSequence}. We shall visit this description again here with more details for convenience. These problems can be described as follows.
\begin{enumerate}
    \item \emph{Local inversion of a map from the data of a partial iterative sequence}. Given a map $F:\ftwo^n\rightarrow\ftwo^n$ local inversion problem is to find $x$ such that for a given $y$ we have $y=F(x)$. We proposed in previous work \cite{Sule1,Sule2,Sule3,Sule4} that this problem can be addressed by the Black Box Linear Algebra by finding the recurrence relations of the iterative sequence $y(k+1)=F(y(k))$ where $y(0)=y$, $k=1,2,\ldots$. The resulting sequence is always ultimately periodic with an unknown and exponential period $N=O(2^n)$. However a practically feasible length sequence $\{y(k)\}$ upto length $M$ is assumed to be available for computation such as (sub-exponential to polynomial size) $M=O(2^{(\log N)^{\alpha}(\log\log N)^{1-\alpha}k})$ where $0\leq\alpha\leq 1$. Hence the local inverse $x$ solved from a partial sequence has a probability of being the true inverse of the original sequence. A separate computational test is made available for verification of the correct inverse. This problem has important application in Cryptanalysis and solving the Observability problem in finite state systems.
    \item \emph{Inversion of a sequence from a partially specified sequence}. As suggested in this paper the local inversion problem of maps can be extended to the problem of inversion of sequences. Above practical issue of availability of the original sequence of limited length is identical in both these problems. There is a sequence $S(N)$ of length $N$ which is exponential. We are given a partial sequence $S(M)$ of a length $M$ which is polynomial to sub-exponential size in $\log N$ and it is required to find the inverse of $S(N)$ using recurrence relations discovered over $S(M)$. Hence the inverse $s_{(-1)}$ has a probability of being a true inverse of the original sequence $S(N)$. A separate computational test is available for verification of the correct inverse.
\end{enumerate}

From the above description of the practical issue in the inversion problems it follows that the complexity of inversion is dependent on the length of the sequence available for computation of inversion. Further the expected probability of the correctness of the solution also depends on the sequence length as well as the order and degree of the associated polynomials. Longer the length $M$ of the available sequence from that of the full sequence, higher is the expected probability of correctness of the solution. Once the length $M$ of the given sequence is bounded, the complexity of solving the inversion problem is also bounded. Hence it is more appropriate to define expected probability of inversion as the probability of correct inversion in the average case of sequences of given partial sequence of length $M$.

\begin{definition}[Expected Probability and Expected Complexity of Inversion] If $N$ is a length of sequences of exponential size and $M$ is of sub-exponential size. Then the \emph{expected complexity of inversion} is the complexity of inversion as a function of $N$ in the average case of sequences of length $M$. Expected probability of inversion is the probability of inverse of a partial sequence of length $M$ to be the inverse of the full sequence of length $N$. 
\end{definition}

In Cryptanalysis as well as in Observability problems of finite state dynamical systems, the algorithm which generates the given sequence is available (or the map with the local point is available) such as PR generators with an IV, encryption algorithms with known plaintext or the maps of system dynamics with an output sequence. The practical issue is then what is the probablity of correct inversion when only a partial sequence of sufficient length is available. The well known method of TMTO attack on local inversion $y=F(x)$ of a map $F$ is a search method of order $O(2^n/2)$ to solve the inversion problem. However this method does not utilise the structure of the underlying finite field and higher properties of sequences such as recurrence relations and the minimal polynomial. The method of computing the associated polynomial from a given sequence to solve the inversion problem has complexity dependent on the order of the associated polynomials and rank of the Hankel matrix of monomials. The product of these two as a complexity measure is the complexity of inversion using recurrence relations of the sequence and the associated polynomials. 

\subsection{Conjecture on probability of correct inversion using a partial sequence}
If $S(M)$ is a partial sequence of an original periodic sequence $S(N)$, then the probability of correct inversion of $S(N)$ using a recurrence relation on $S(M)$ is the same as the probability of an associated polynomial of $S(M)$ having an inverse to be an associated polynomial of $S(N)$. Hence, for a fixed set of monomials ${\cal M}$, if the rank of the Hankel matrix $r(M)=\rank H({\cal M}(S(M))$ is maximal for the sequence $S(N)$ i.e. 
\beq\label{MaxRankrelation}
r({\cal M})=\rank H({\cal M}(S(N))
\eeq
then the inverse of $S(M)$ obtained over ${\cal M}$ is expected to be correct inverse of $S(N)$. Hence the expected complexity of inversion is the complexity of inversion with respect to a chosen set ${\cal M}$ which is $C({\cal M})=\sqrt{mr({\cal M})}$. 

Since there are multiple solutions of associated polynomials defined by linear forms over ${\cal M}$ but for each subset ${\cal M}_1\subset {\cal M}$ whose Hankel matrix has the same maximal rank equal to $|{\cal M}_1|$ the number of monomials in ${\cal M}_1$, there is a unique associated polynomial defined over ${\cal M}_1$. Hence it is logical to conjecture the following.

    \begin{conjecture} Let $S(M)$ be a partial sequence of an original sequence $S(N)$, $N>>M$. Let ${\cal M}$ be a fixed set of monomials and $r({\cal M})$ be the rank of $H({\cal M}(S(M))$. If the monomial set ${\cal M}$ is sufficient to capture the maximal rank of the sequence $S(N)$ i.e. the condition (\ref{MaxRankrelation}) holds, then the probability of correct inversion using the partial sequence $S(M)$ and associated polynomials defined over ${\cal M}$ is
    \[
    P_{\cal M}=\frac{1}{
    \mbox{\rm Number of subsets ${\cal K}\subset{\cal M}$ with maximal rank equal to $r({\cal M})$}
    }
    \]
    This is the expected probability of inversion of $S(N)$ using the partial sequence $S(M)$ relative to a monomial set ${\cal M}$ satisfying (\ref{MaxRankrelation}) 
    \[
    P_{\rm exp}=P_{\cal M}. 
    \]
    However this leaves open the question of probability of the condition (\ref{MaxRankrelation}) being satisfied for a monomial set ${\cal M}$ which achieves maximal rank at $S(M)$. Let this be denoted as $P_0$. The final probability of correct inversion for any ${\cal M}$ which achieves maximal rank at $S(M)$ is then
     \[
     P_{\rm Inv}=P_{\rm exp}P_0
     \]
 \end{conjecture}

 \subsection{Consequences of previous known estimates of nonlinear complexity}
 Non-linear complexity of sequences was proposed by Jansen \cite{JansenBoe,Jansen91} as the shortest order of feedback shift register (FSR) which can generate the sequence through the recurrence relations. Hence this complexity is the same as the shortest order $m$ of an associated polynomial of $S(M)$ considered in this paper. For invertibility of the sequence using the RRs defined by such a polynomial we also hypothesised maximal rank being achieved by the Hankel matrix of evaluation of the associated polynomial of order $m$ and degree $d$ at $S(M)$ and solvability of the invertibility relation (\ref{InverseRelationwithcoeff}). 
 
 One of the important result of Jansen whose proof has been revisited recently in \cite{Yuanetal} is that the nonlinear complexity of random sequences $S(N)$ on an average is of the order $m=2\log N$. Hence the expected order $m$ for inversion of $S(N)$ using a monomial set ${\cal M}$ is $2\log N$ and the average complexity of inversion using such an ${\cal M}$ is $(2\log N)r({\cal M})$ where the rank of the Hankel matrix $r({\cal M})=H({\cal M})(S(M))$ is of polynomial order $O((\log N)^k)$ when the number of monomials in ${\cal M}$ is of polynomial order.  Hence the above conjecture has the concrete description as follows. Note that the number of subsets of ${\cal M}$ giving maximal rank of the Hankel matrix over $S(M)$ is also of polynomial order $O((\log N)^k)$. 

 \begin{conjecture}
     The expected probability of correct inversion of $S(N)$ when the partially given sequence $S(M)$ satisfies the condition (\ref{MaxRankrelation}) for a monomial set ${\cal M}$ of order $m=2\log N$ and number of monomials sufficient to attain the maximal rank of the Hankel matrix over $S(M)$ is
     \[
     P_{\rm exp}=\frac{1}{\mbox{\rm Number of subsets ${\cal K}$ of ${\cal M}$ with maximal rank $r({\cal M})=|{\cal K}|$}}
     \]
     and the average case complexity of inversion is
     \[
     C_{\rm Inv}=O((\log N)^{k})
     \]
     However this leaves open the question of probability of the condition (\ref{MaxRankrelation}) being satisfied for a monomial set ${\cal M}$ which achieves maximal rank at $S(M)$. Let this be denoted as $P_0$. The final probability of correct inversion for any ${\cal M}$ which achieves maximal rank at $S(M)$ is then
     \[
     P_{\rm Inv}=P_{\rm exp}P_0
     \]
 \end{conjecture}

In realistic problems the index $k$ is likely to be small and the resulting process of computation of inverse is likely to be practically feasible. Evidence for the above conjectures needs to be gathered from case studies of sequences arising from various problems in Cryptography and Finite state systems. 

\subsubsection{Inversion using \PCId}
Some of the problems that can in principle be solved by inversion of sequences (as local inversion of maps) have been discussed in the previous articles \cite{Sule2,Sule3,Sule4}. It is worth recalling these challenge problems here. For instance local inversion of maps can in principle solve, RSA encryption without factoring the modulus, breaking RSA by finding equivalent private key without factoring the modulus, discrete log problems in finite fields and elliptic curves, inversion of One Way Functions, key recovery problems under known plaintext attack from encryption (for both block and stream ciphers), seed recovery from random number generators etc. In previous work mentioned above these problems were theoretically addressed by inversion using LC. However average case LC is exponential. Hence complexity of inversion using degree $1$ associated polynomials is exponential. 

On the other hand as shown in \cite{Yuanetal} the MOC is of polynomial order in the period of the original sequence on the average. Hence whenever the condition (\ref{MaxRankrelation}) holds for a sequence with small enough degree $d$, a monomial set of polynomial size can solve the inversion problem even for an exponentially long sequence with high probability close to $0.5$. We have made conjectures concerning the complexity and probability of correct inversion above based on this premise that the complexity of inversion of a sequence is of polynomial order for higher degree. Hence \PCId\ may resolve the challenge problems referred above in average cases in feasible time. This verification is however beyond the scope of the present paper.

%% file: Conclusions.tex
\section{Conclusion}
A theory of complexity of inversion of sequences is developed using non-linear (polynomial) functions instead of the well known linear functions which refer to the linear complexity of sequences. Hence the complexity is defined as the minimal number of variables required in polynomials at a fixed degree which satisfy recurrence relations of the sequence called as associated polynomials. The complexity of inversion is referred to the number of variables at a fixed degree when a unique inverse can be solved from the associated polynomials. Conditions for existence of such polynomials and solvability of the inverse are described in terms of solutions of a linear system of equations and a matrix of these equations called as the Hankel matrix. The Hankel matrix is a matrix of evaluation of monomials in associated polynomials with largest degree equal to the a-priori fixed degree. This matrix is a generalisation of the Hankel matrix well known in linear complexity (LC) theory of sequences and refers to polynomial recurrence relations of a fixed degree. Inversion of the sequence in terms of LC is shown as a special case. The conditions of existence of associated polynomials having inverse is also extended to an a-priori chosen set of monomials which appear linearly in associated polynomials of the sequence as well for polynomials with special structure. In conclusion this paper develops a nonlinear generalisation of the LC based theory of complexity and inversion of sequences with a major advantage that while the average case of LC of inversion is exponential time, the average case of polynomial complexity of inversion at a fixed degree is conjectured to be polynomial time. In the final section it is conjectured that the average case polynomial complexity of inversion at a fixed degree of random sequences is of polynomial order and hence once the bound on the complexity is predicted, the computation of inversion of sequences is possible in polynomial time in the average case. This conjecture needs to be verified and established for case studies of sequences arising in realistic problems which is a massive computational task beyond the scope of the present paper. However, if the conjecture turns out to have positive empirical evidence the nonlinear complexity of inversion shall be a disruptive tool in solving many problems of computation which are perceived as difficult.   

%% file: main.bbl
\begin{thebibliography}{9}
\bibitem{Nied} H. Niederreiter. Linear Complexity and Related Complexity Measures for Sequences. INDOCRYPT 2003, LNCS 2904, pp.1-17, 2003.
\bibitem{Rueppel} R. A. Rueppel. New Approaches to Stream Ciphers, PhD. Thesis, Swiss Federal Institute of Technology, Zurich, 1984.
\bibitem{Golomb} S. W. Golomb. Shift Register Sequences, HoIden-Day Inc., San Francisco, 1967.
\bibitem{Jansen} C. J. A. Jansen. Investigations On Nonlinear Stream Cipher Systems:
Construction and Evaluation Methods. Ph.D. Thesis, Delft University, 1989.
\bibitem{Jansen91} C.J.A. Jansen. Maximal order complexity of ensemble sequences. Advances in Cryptology - EUROCRYPT '91, LNCS 547, pp. 153-159,
Springer-Verlag, Berlin Heidelberg, 1991.
\bibitem{JansenBoe} C. J. A. Jansen and D. E. Boekee. The shortest feedback shift register that can generate the given sequence. Advances in Cryptology - CRYPT0 ‘89, LNCS 435, pp. 90-99, Springer-Verlag, Berlin Heidelberg, 1990
\bibitem{LiuML} L. Liu 1, S. Miao 2 and B. Liu. On Nonlinear Complexity and Shannon’s Entropy of Finite Length Random Sequences. Entropy, vol.17, pp.1936-1945, 2015. doi:10.3390/e17041936.
\bibitem{LeylaWinterhof} I. Leyla, Arne Winterhof. Maximum-order Complexity and Correlation Measures. arxiv.org/math.NT/1703.09151, March 2017.
(published paper: I. Leyla and A. Winterhof, “Maximum-order complexity and correlation measures”, Cryptography, vol. 1, no. 1, pp. 1–7, May 2017).
\bibitem{ChenGomez} Z. Chen, I. Gómez, D. Gómez-Pérez, and A. Tirkel, “Correlation measure, linear
complexity and maximum order complexity for families of binary sequences,” Finite
Fields Appl., vol. 78, no. 101977, Feb. 2022.
\bibitem{Yuanetal} Qin Yuan, Chunlei Li, Xiangyong Zeng, Tor Helleseth, and Debiao He. Further Investigations on Nonlinear Complexity of Periodic Binary Sequences. arxiv.org/2404.16313v1/cs.IT, April 2024.
\bibitem{Blahut} R. E. Blahut. Cryptography and Secure Communication. Cambridge University Press, 2014.
\bibitem{Golomb2} S. W. Golong, G. Gong. Signal Design for Good Correlation for Wireless Communication, Cryptography and Radar. Cambridge University Press, 2005.
\bibitem{Hellman} Martin E. Hellman. A cryptanalytic time-memory trade-off. IEEE Trans. Information Theory, IT 26(4):401–406, 1980
\bibitem{HongSarkar} Jin Hong and Palash Sarkar. New applications of time memory data tradeoffs. Advances in Cryptology, ASIACRYPT 2005, LNCS vol.3788, pp.353–372. Springer, 2005.
\bibitem{Hays} H.\ M.\ Hays. Distributed Time Memory Tradeoffs.\\ http://eprint.iacr.org/2018/123.
\bibitem{BiryukovShamir}Alex Biryukov and Adi Shamir. Cryptanalytic time/memory/data tradeoffs for stream ciphers. Advances in Cryptology, ASIACRYPT 2000, LNCS vol.1976, pp.1–13. Springer, 2000
\bibitem{Sule4} Virendra Sule. Local inversion of maps: Black box cryptanalysis. (Invited article). Computer Algebra Magazine, CA-Rundbrief 71 (2022) vol.27, October 22.
\bibitem{Sule3} Virendra Sule: Local inversion of maps: Black Box Cryptanalysis. arXiv.org/cs.CR/2207.03247v2. July, 2022.
\bibitem{Sule2} Virendra Sule. Local Inversion of Maps: A New attack on Symmetric Encryption, RSA and ECDLP. http://arxiv.org/cs.CR/2202.06584v2. March 2022
\bibitem{Sule1} Virendra Sule: A complete algorithm for local inversion of maps: Application to Cryptanalysis. arXix.org/cs.CR/2202.06584v2. May, 2021.
\bibitem{Peng} Jiarong Peng, Xiangyong Zeng and Zhimin Sun. Finite length sequences with large nonlinear complexity. AIMS 2018, Vol.12, issue 1, pp.215-230. Doi: 10.3934/amc.2018015 AIMS (American Institute of Mathematical Sciences)
\bibitem{Liuetal} Lingfeng Liu 1, Hongyue Xiang, Renzhi Li and Hanping Hu. The Eigenvalue Complexity of Sequences in the Real Domain,
Entropy 2019, 21, 1194; doi:10.3390/e21121194
\bibitem{Nied21} Zhixiong Chen, Ana I. Gomez, Domingo Gomez-Perez, Laszlo Merai, Harald Niederreiter. On the expansion complexity of sequences over finite fields
 arxiv.org/1702.05329, Feb 2017.
\bibitem{LempelZiv}A. Lempel and J. Ziv. “On the Complexity of Finite Sequences”, IEEE Trans.on Info. Theory, vol. IT-22, no. 1, pp. 75-81, January 1976
\bibitem{Massey} J. L. Massey. “Shift-Register Synthesis and BCH Decoding”, IEEE Trans. on Information Theory, vol. IT-15, January 1969.
\end{thebibliography}
